\documentclass[11pt]{elsarticle}
\usepackage{a4wide}

\newcommand{\ignore}[1]{}

\journal{Discrete Applied Mathematics}

\usepackage{amssymb}
\usepackage{amsfonts}
\usepackage{microtype}
\usepackage{amsmath}
\usepackage{amsthm}
\usepackage{booktabs} 
\usepackage{color}
\usepackage{microtype}
\usepackage[mathcal,mathscr]{euscript} 
\usepackage{relsize} 
\usepackage[pdftex,all,arc,curve,color,line]{xy}

\usepackage[ansinew]{inputenc}
\usepackage[T1]{fontenc}
\usepackage{type1cm}

\usepackage{newtxtext}
\usepackage[varvw]{newtxmath}

\usepackage[bb=boondox]{mathalfa}

\usepackage{enumitem}
\setitemize{topsep=2pt,itemsep=-1ex,partopsep=1ex,parsep=1ex}
\setenumerate{topsep=2pt,itemsep=-1ex,partopsep=1ex,parsep=1ex}

\let\OLDthebibliography\thebibliography
\renewcommand\thebibliography[1]{
  \OLDthebibliography{#1}
  \setlength{\parskip}{0pt}
  \setlength{\itemsep}{0pt plus 0.3ex} }

\usepackage[colorlinks]{hyperref}
\definecolor{gruen}{rgb}{0,.5,0}
\hypersetup{colorlinks=true,linkcolor=gruen,citecolor=gruen}

\usepackage[capitalise,nameinlink]{cleveref}

\crefname{appendix}{}{}
\crefname{lem}{Lemma}{Lemmas}
\crefname{thm}{Theorem}{Theorems}
\crefname{rem}{Remark}{Remarks}
\crefname{ex}{Example}{Examples}
\crefname{probl}{Problem}{Problems}
\crefname{clm}{Claim}{Claims}
\crefname{prop}{Proposition}{Propositions}
\crefname{figure}{Figure}{Figures}
\crefname{fact}{Fact}{Facts}
\crefname{obs}{Observation}{Observations}

\newtheorem{thm}{Theorem}
\newtheorem{lem}{Lemma}
\newtheorem*{hlem}{Homogenization Lemma}
\newtheorem*{envlem}{Envelope Lemma}
\newtheorem*{decomplem}{Decomposition Lemma}
\newtheorem{fact}{Fact}

\newtheorem{probl}{Problem}

\theoremstyle{definition}
\newtheorem{dfn}{Definition}

\theoremstyle{remark}
\newtheorem{rem}{Remark}
\newtheorem{ex}{Example}
\newtheorem*{notation}{Notation}

\def\qed{\mbox{}\hfill\vbox{\hrule height0.6pt\hbox{%
  \vrule height1.3ex width0.6pt\hskip0.8ex
   \vrule width0.6pt}\hrule height0.6pt
  }}

\def\deg#1{\mathrm{deg}(#1)}
\let\geq\geqslant
\let\leq\leqslant
\let\daug\odot
\let\suma\oplus
\let\bar\overline
\newcommand{\plus}{\mbox{\footnotesize $+$}}
\newcommand{\NN}{\mathbb{N}}
\newcommand{\RR}{\mathbb{R}}

\newcommand{\f}{\mathcal{F}}
\newcommand{\pRR}{\mathbb{R}_{\plus}}
\newcommand{\Exp}[1]{B_{#1}}
\newcommand{\Low}[1]{A_{#1}}
\newcommand{\pos}[1]{{#1}^+}
\newcommand{\supp}[1]{\mathrm{sup}(#1)}
\newcommand{\Supp}[1]{\mathrm{Sup}(#1)}
\newcommand{\const}[1]{c_{#1}}
\newcommand{\Up}[1]{#1^{\uparrow}}
\newcommand{\lenv}[1]{\lfloor #1\rfloor}
\newcommand{\arithm}[1]{\mathsf{Arith}(#1)}
\newcommand{\BBool}[1]{\mathsf{B}(#1)}
\newcommand{\linBool}[1]{\mathsf{B}_{\mathrm{lin}}(#1)}
\newcommand{\mlinBool}[1]{\mathsf{B}^{\plus}_{\mathrm{lin}}(#1)}
\newcommand{\Bool}[2]{\mathsf{B}_{#2}(#1)}
\newcommand{\sBool}[2]{\mathsf{B}^{\ast}_{#2}(#1)}
\newcommand{\BP}[2]{\mathsf{BP}_{#2}(#1)}

\newcommand{\kMin}[2]{\mathsf{Min}_{#2}(#1)}
\newcommand{\perm}[1]{\mathrm{Per}_{#1}}
\newcommand{\match}[1]{\mathrm{Match}_{#1}}
\newcommand{\cov}[1]{\mathrm{Cov}_{#1}}
\newcommand{\lines}[1]{\mathrm{Lines}_{#1}}
\newcommand{\skal}[1]{\langle #1\rangle}
\newcommand{\Item}[1]{\item[\mbox{\rm (#1)}]} 

\newcommand{\vienas}{\mathbb{1}}%
\newcommand{\cf}[1]{{#1}^{o}}
\newcommand{\up}[1]{#1^{\triangledown}}
\DeclareRobustCommand{\bigO}{%
  \text{\usefont{OMS}{cmsy}{m}{n}O}}
\newcommand{\mon}[1]{\mathrm{mon}(#1)}
\newcommand{\mdeg}[1]{\mathrm{dg}(#1)} 
\newcommand{\dual}[1]{{#1}^{\ast}}

\newcommand{\Ln}{{\cal L}} 
\newcommand{\ml}[1]{\mathlarger{#1}}
\newcommand{\T}[1]{T(#1)} 
\newcommand{\Hom}[2]{{#1}[#2]}
\newcommand{\PI}[1]{\mathrm{PI}(#1)}

\def\markov{
\begin{figure}[t]
   \[
   \xygraph {
   !{<0cm,0cm>;<1.5cm,0cm>:<0cm,-1.5cm>::}
    !{(-0.2,0)}*+={s}="s" !{(0,0)}*{\bullet}="a1"
    !{(1,0)}*{\circ}="a2" !{(2,0)}*{\circ}="a3"
    !{(2.5,0)}*{\cdots}="dots" !{(3,0)}*{\circ}="a4"
    !{(4,0)}*{\circ}="a5"
    !{(0,0.5)}*{\circ}="b1" !{(1,0.5)}*{\circ}="b2"
    !{(2,0.5)}*{\circ}="b3" !{(2.5,0.5)}*{\cdots}="dots"
    !{(3,0.5)}*{\circ}="b4" !{(4,0.5)}*{\circ}="b5"
    !{(0,1)}*{\circ}="c1" !{(1,1)}*{\circ}="c2" !{(2,1)}*{\circ}="c3"
    !{(2.5,1)}*{\cdots}="dots" !{(3,1)}*{\circ}="c4"
    !{(4,1)}*{\circ}="c5"
    !{(0,1.3)}*{\vdots}="d1" !{(1,1.3)}*{\vdots}="d2"
    !{(2,1.3)}*{\vdots}="d3" !{(3,1.3)}*{\vdots}="d4"
    !{(4,1.3)}*{\vdots}="d5"
    !{(0,1.7)}*{\circ}="e1" !{(1,1.7)}*{\circ}="e2"
    !{(2,1.7)}*{\circ}="e3" !{(2.5,1.7)}*{\cdots}="dots"
    !{(3,1.7)}*{\circ}="e4" !{(4,1.7)}*{\bullet}="e5"
    !{(4.2,1.7)}*+={t}="t"
    "a1":"a2"^{\ml{x_1}} "a2":"a3"^{\ml{x_2}} "a4":"a5"^{\ml{x_r}}
    "b1":"b2"^{\ml{x_2}} "b2":"b3"^{\ml{x_3}} "b4":"b5"^{\ml{x_{r+1}}}
    "c1":"c2"^{\ml{x_3}} "c2":"c3"^{\ml{x_4}} "c4":"c5"^{\ml{x_{r+2}}}
    "a1":"b1" "a2":"b2" "a3":"b3" "a4":"b4" "a5":"b5"
    "b1":"c1" "b2":"c2" "b3":"c3" "b4":"c4" "b5":"c5"
    "e1":"e2"^{\ml{x_{n-r+1}}} "e2":"e3"^{\ml{x_{n-r+2}}} "e4":"e5"^{\ml{x_n}} }
  \]
  \caption{\footnotesize A monotone read-$1$ branching program
    computing the threshold-$r$ function $\mathrm{Th}^n_r(x)=1$ iff
    $x_1+x_2+\cdots+x_n\geq r$. Unlabeled edges are rectifiers (are
    labeled by constant $1$). This BP is even a \emph{syntactically}
    read-$1$ BP, and has $r(n-r+1)$ switches. On the other hand,
    Markov~\cite{markov} has shown that \emph{every} monotone BP for
    $\mathrm{Th}^n_r$ must have at least this number $r(n-r+1)$ of
    switches.  Thus, at least for $k=1$, monotone read-$k$ BPs
   \emph{can} be optimal among all monotone BPs.}
  \label{fig:markov}
\end{figure}
}

\begin{document}

\begin{frontmatter}

  \title{Notes on Boolean Read-k and Multilinear Circuits}

  \author{Stasys Jukna}
  \ead{stjukna@gmail.com}
  \ead[url]{https://web.vu.lt/mif/s.jukna/}

  \affiliation{organization={Faculty of Mathematics and Computer
      Science, Vilnius University}, city={Vilnius},
    country={Lithuania}}

  \begin{abstract}
    A monotone Boolean $(\lor,\land)$ circuit computing a monotone
    Boolean function $f$ is a read-$k$ circuit if the polynomial
    produced (purely syntactically) by the arithmetic $(+,\times)$
    version of the circuit has the property that for every prime
    implicant of $f$, the polynomial contains at least one monomial
    with the same set of variables, each appearing with degree $\leq
    k$. Every monotone circuit is a read-$k$ circuit for some~$k$.  We
    show that already read-1 $(\lor,\land)$ circuits are not weaker than
   monotone arithmetic constant-free  $(+,\times)$ circuits computing
    multilinear polynomials, are not weaker than non-monotone
    multilinear $(\lor,\land,\neg)$ circuits computing monotone
   Boolean functions, and have the same power as tropical $(\min,+)$ circuits solving $0/1$ minimization
    problems.
     Finally, we show that read-2
    $(\lor,\land)$ circuits can be exponentially smaller than read-1
    $(\lor,\land)$ circuits.
  \end{abstract}


\end{frontmatter}

\section{Introduction}
\label{sec:intro}

Proving lower bounds on the size of arithmetic $(+,\times,-)$ circuits
as well as on the size of Boolean $(\lor,\land,\neg)$ circuits remains
a notoriously hard problem.  Although the problem has received a great
deal of attention for decades, the best known lower bounds for
arithmetic circuits computing explicit multilinear $n$-variate polynomials
with $0$-$1$ coefficients  remain barely super-linear bounds $\Omega(n\log n)$ proved by Baur and Strassen~\cite{Baur1983} already in 1983. For Boolean
circuits, known lower bounds are even not super-linear.
The first lower bound $2n$ was proved by Schnorr~\cite{schnorr74} in 1974,
and improved to $3n$ by Blum~\cite{blum84} in 1984;
the best known lower bound $4.5n-o(n)$ was proved by Lachish and Raz~\cite{LachishR01} in 2001, and  improved to $5n-o(n)$ by Iwama and Morizumi~\cite{IwamaM02}.
In both circuit models (arithmetic and Boolean), super-polynomial lower bounds are known only for restricted circuits  such as bounded-depth and monotone circuits. The books \cite{BCS97,myBFC-book,wegener} and  recent surveys
\cite{ChenKW11,Shpilka2010} provide wide coverage of  Boolean and arithmetic circuits.

This lack of proofs of strong lower bounds for unrestricted arithmetic and
Boolean circuits happens mainly because such circuits can use
cancellations $x-x=0$ in the arithmetic, and can use cancellations
$x\land\bar{x}=0$ in the Boolean case. Understanding the role of
cancellations in arithmetic and Boolean circuits remains the ultimate
goal of circuit complexity.

Monotone arithmetic $(+,\times)$ circuits cannot use cancellations $x-x=0$, while
monotone Boolean $(\lor,\land)$ circuits cannot use cancellations
$x\land\bar{x}=0$. Still, the task of proving lower bounds even for
monotone \emph{Boolean} circuits turned out to be much more difficult
than that for monotone \emph{arithmetic} circuits.  Although super-polynomial lower bounds on the size
of monotone arithmetic $(+,\times)$ circuits
were known starting with the notable paper by
Schnorr~\cite{schnorr} of 1976, it took a decade until such lower
bounds for monotone Boolean $(\lor,\land)$ circuits were proved by
Razborov~\cite{razb-clique,razb-perm} in 1985; until
then, the best known lower bound on the size of monotone Boolean circuits was only
$4n$ proved by Tiekenheinrich~\cite{Tiekenheinrich84} in 1984.  This happens because Boolean circuits can use
idempotence laws $x\lor x=x$ and $x\land x=x$ as well as the
absorption law $x\lor xy=x$, while arithmetic circuits cannot use any
of these laws.  The current paper attempts to identify a possible
source for this discrepancy: the presence of multiplicative
idempotence and absorption in Boolean circuits.

It turned out that the absence of \emph{additive} idempotence $x\lor
x=x$ in arithmetic circuits (where $x+x\neq x$) is not a crucial
issue: most  lower bounds (albeit not all, \cite{Yehudayoff19} being a nice exception)  on the monotone
arithmetic $(+,\times)$ circuit complexity,
including~\cite{gashkov,GS12,jerrum,juk-SIDMA,Raz2011,
schnorr,shamir1980,tiwari1994,valiant80}
were proved by only using the \emph{structure} of monomials and fully
ignoring actual values of their (nonzero) coefficients.

But the absence of \emph{multiplicative} idempotence $x\land x=x$ and
absorption $x\lor xy=x$ in the arithmetic world turned out to be
crucial even in the case of monotone circuits. The goal of this
article is to show that already a \emph{very restricted} use of
multiplicative idempotence, in combination with
absorption, makes a big difference between Boolean and arithmetic
circuits.

To fine grain the ``degree'' of multiplicative idempotence, and by
analogy with read-$k$ branching programs, we introduce (in
\cref{sec:readk}) the so-called ``read-$k$'' $(\lor,\land)$ circuits.
Our goal is to show that already read-$1$ $(\lor,\land)$ circuits
capture the power of three different types of circuits: monotone
\emph{arithmetic} $(+,\times)$ circuits, Boolean \emph{multilinear} DeMorgan
$(\lor,\land,\neg)$ circuits, as well as \emph{tropical}\footnote{
The adjective ``tropical'' was coined by French mathematicians in honor of Imre Simon who lived in Sao Paulo (south tropic).  Tropical algebra and tropical geometry are now intensively studied topics in
mathematics.} $(\min,+)$
circuits.

The latter model of (tropical) circuits is motivated by
dynamic programming (DP).  Namely, many classical DP algorithms for
minimization problems are ``pure'' in that they only use $(\min,+)$
operations in their recursion equations. Prominent examples of pure DP
algorithms are the Bell\-man--Ford--Moore shortest $s$-$t$ path
algorithm, the Roy--Floyd--Warshall all-pairs shortest paths
algorithm, the Bellman--Held--Karp traveling salesman algorithm, the
Dreyfus--Levin--Wagner Steiner tree algorithm, and many others.  On
the other hand, pure DP algorithms are just \emph{special}
(recursively constructed) tropical $(\min,+)$ circuits.
Thus, any lower bound on the size of $(\min,+)$ circuits is also a lower bound on the minimum number of operations that any pure DP algorithm solving a given minimization problem must perform, be the designer of an algorithm even omnipotent.

First lower
bounds on the size of tropical circuits were proved already decades
ago, including Kerr~\cite{Kerr1970}, Jerrum and Snir~\cite{jerrum}, as well as recently,
including Grigoriev and Koshevoy~\cite{GrigorievK16}, Grigoriev and Podolskii~\cite{GrigorievP20}, Mahajan, Nimbhorkar
and Tawari~\cite{MahajanNT17,MahajanNT19},
Jukna and Seiwert~\cite{juk-SIDMA,JS20a}. In fact, as shown by Jerrum and Snir~\cite{jerrum},
if an arithmetic polynomial $P$ is multilinear
and homogeneous, then the $(\min,+)$ circuit complexity of the corresponding minimization problem is at least the monotone arithmetic $(+,\times)$ circuit complexity of the polynomial $P$. Thus, many other lower bounds for tropical $(\min,+)$ circuits follow from the aforementioned earlier lower bounds on the monotone arithmetic $(+,\times)$ circuit complexity of the corresponding polynomials, including bounds shown many years ago by
Schnorr~\cite{schnorr}, Shamir and Snir~\cite{shamir1980},
Valiant~\cite{valiant80}, Tiwari and Tompa~\cite{tiwari1994} and other authors.

\section{Results}
\label{sec:results}

The model of read-$k$ circuits is quite natural and is by
 analogy with the intensively investigated model of
\emph{read-$k$ branching programs}.  Intuitively, ``read-$k$'' means 
that one cannot ``benefit'' from (multiplicative) idempotence for more than $k$ times.

Let $F$ be a monotone Boolean $(\lor,\land)$ circuit. Throughout the article, we assume that constants $0$ and $1$ are not are not used as inputs in Boolean circuits: such inputs are not necessary when computing non-trivial (non-constant) Boolean functions.
An \emph{arithmetic version} of $F$ is a monotone arithmetic
$(+,\times)$ circuit obtained from $F$ by replacing every OR gate
with an addition gate, and every AND gate with a multiplication gate.
The obtained arithmetic circuit produces (purely syntactically) some polynomial $P_F$, which we call the \emph{formal polynomial} of the Boolean circuit $F$. A \emph{shadow monomial} of a Boolean term $t=\bigwedge_{i\in S}x_i$  is a monomial $p=\prod_{i\in S}x_i^{d_i}$ with the same set of variables as $t$ and all degrees $d_i\geq 1$.
For example, $x^2y^3$ is a shadow monomial of $xy$.
It is easy to show (see \cref{lem:struct-bool}) that
the circuit $F$ computes a monotone Boolean function $f:\{0,1\}^n\to\{0,1\}$ if and only if
the formal polynomial $P_F$ of the circuit $F$ has the following two properties:
\begin{itemize}
\item[(i)] every monomial of $P_F$ contains all variables of at least
  one prime implicant of $f$;

\item[(ii)] every prime implicant of $f$ has at least one shadow
  monomial in $P_F$.
\end{itemize}
Intuitively, (i) reflects the absorption property $x\lor xy=x$ (longer monomials ``do not matter''), while
(ii) reflects the multiplicative idempotence $x\land x=x$ (large degree ``does not matter'').  In
\emph{read}-$k$ circuits we strengthen the property (ii) and require that
every prime implicant of $f$ has at least one shadow monomial in $P_F$
in which each variable appears with degree $\leq k$. There are no
restrictions on the degrees of other monomials of~$P_F$.

\begin{ex}\label{ex:readk}
  The Boolean $(\lor,\land)$ circuit $F=(x\lor
  y)(x\lor z)(y\lor z)$ computes the Boolean function
  $f(x,y,z)=1$ iff $x+y+z\geq 2$ (the majority function of three
  variables). The arithmetic $(+,\times)$ version of $F$ is the
  circuit $F'=(x+y)(x+z)(y+z)$, and the polynomial produced by it is
  $P_F=x^2y + xy^2 + x^2z + y^2z + xz^2 +yz^2 +2xyz$.  The (Boolean)
  circuit $F$ is a read-$2$ but not a read-$1$ circuit, because, for
  example, both shadow monomials $x^2y$ and $xy^2$ of the prime
  implicant $xy$ of $f$ in the polynomial $P_F$ have a variable of degree
  $>1$.  Note,
  however, that the Boolean circuit $H=xy\lor xz\lor
  yz$ (the dual of the circuit $F$) also computes $f$ but already is a read-$1$ circuit: the
  corresponding to this circuit polynomial is $P_H=xy+xz+yz$.  \qed
\end{ex}

Our main results are the following. A \emph{DeMorgan circuits} is a
Boolean  $(\lor,\land,\neg)$ circuit with negations only
applied to input variables. That is, such a circuit is a monotone $(\lor,\land)$ circuit whose inputs are variables $x_1,\ldots,x_n$
and their negations $\bar{x}_1,\ldots,\bar{x}_n$.
A DeMorgan $(\lor,\land,\neg)$
circuit is \emph{multilinear} if the two Boolean functions $g$ and $h$ computed at
the inputs to any AND gate depend on disjoint sets of variables. For
example, the functions $g=x\lor xy$ and $h=\bar{y}\lor x\bar{y}$
depend on disjoint sets of variables: $g$ depends only on $x$, while
$h$ depends only on~$y$. Every family  $\f\subseteq 2^{[n]}$ of subsets of
$[n]:=\{1,2,\ldots,n\}$ defines a multilinear $n$-variate polynomial $P_{\f}(x):=\sum_{S\in \f}\prod_{i\in S} x_i$. A polynomial $Q$ is
 \emph{similar} to the polynomial $P_{\f}$ if it is of the form
$Q(x)=\sum_{S\in \f}\const{S}\prod_{i\in S}x_i$
for
some integer coefficients $\const{S}\geq 1$; in particular, the polynomial $P_{\f}$ is similar to itself (then all $\const{S}=1$). An arithmetic $(+,\times)$ circuit is \emph{constant-free} if it has no constants among the inputs.
We prove the following, where $\f\subseteq 2^{[n]}$ is an arbitrary antichain (no two sets of $\f$ are comparable under the inclusion).

\begin{enumerate}
\item Read-$1$ $(\lor,\land)$ circuits are \emph{not weaker} than
monotone \emph{arithmetic} $(+,\times)$ circuits computing multilinear polynomials in the following sense: if a monotone arithmetic constant-free  $(+,\times)$ circuit computes a  polynomial similar to $P_{\f}$, then
a read-$1$ $(\lor,\land)$ circuit of the same size computes
the Boolean version $f(x)=\bigvee_{S\in\f}\bigwedge_{i\in S}x_i$ of $P_{\f}$. If the polynomial $P_{\f}$ is homogeneous (all sets of $\f$ have the same number of elements), then
the minimum size of a monotone  arithmetic constant-free $(+,\times)$ circuit
computing a polynomial similar to $P_{\f}$ even \emph{coincides} with the
minimum size of a read-$1$ $(\lor,\land)$ circuit computing~$f$ (\cref{thm:envel}).

\item Read-$1$ $(\lor,\land)$ circuits have the \emph{same} power as
  \emph{tropical} $(\min,+)$ circuits in the following sense: the
  minimum size of a $(\min,+)$ circuit computing the tropical polynomial
  $P(x)=\min_{S\in \f}\sum_{i\in S}x_i$ \emph{coincides} with the
  minimum size of a read-$1$ $(\lor,\land)$ circuit computing the Boolean version
  $f(x)=\bigvee_{S\in\f}\bigwedge_{i\in S}x_i$ of $P$
  (\cref{thm:trop-read-once}).

\item Read-$1$ $(\lor,\land)$ circuits are \emph{not weaker} than
  multilinear DeMorgan $(\lor,\land,\neg)$ circuits in the following
  sense: if a multilinear $(\lor,\land,\neg)$ circuit computes a
  Boolean function $f(x)$, then a read-$1$ $(\lor,\land)$ circuit of
  the same size computes the monotone function
  $\up{f}(x):=\bigvee_{z\leq x}f(z)$ (\cref{thm:multilinear}); note
  that $\up{f}=f$ if $f$ is monotone. If the function $f$ is monotone
  and homogeneous (all prime implicants of $f$ have the same number of
  variables), then the minimum size of a multilinear
  $(\lor,\land,\neg)$ circuit computing $f$ even \emph{coincides} with the
  minimum size of a read-$1$ $(\lor,\land)$ circuit computing~$f$
  (\cref{thm:mon-multilin}).

\item Already read-$2$ $(\lor,\land)$ circuits can be exponentially
  smaller than read-$1$ $(\lor,\land)$ circuits and, hence,
  exponentially smaller than tropical $(\min,+)$, monotone arithmetic
  $(+,\times)$, and multilinear $(\lor,\land,\neg)$
  circuits~(\cref{lem:gap}).
\end{enumerate}

\paragraph{Organization} In the preliminary
\cref{sec:preliminaries}, we recall one simple but important
concept---the set of exponent vectors ``produced'' (purely
syntactically) by a circuit over any semiring. Read-$k$ circuits are
introduced in \cref{sec:readk}. The aforementioned relation of
read-$1$ circuits to monotone arithmetic circuits is established in
\cref{sec:arithmetic}. In \cref{sec:explicit} we recall one relatively
simple argument to show large lower bounds for monotone arithmetic
$(+,\times)$ circuits; this is only aimed to demonstrate that the absence
of idempotence and absorption in such circuits is indeed a severe
restriction.  The relation of read-$1$ circuits to tropical circuits
is established in \cref{sec:tropical}, and the relation of read-$1$
circuits to multilinear DeMorgan $(\lor,\land,\neg)$ circuits is
established in \cref{sec:multilin}.  An exponential gap between
read-$1$ and read-$2$ circuits is shown in \cref{sec:gaps}.  In the concluding \cref{sec:conclusions}, several
open problems are formulated. \Cref{app:blocking-lines} contains
a construction of so-called ``blocking lines'' functions as possible candidates to attack these problems.
All proofs are fairly simple.

\section{Preliminaries}
\label{sec:preliminaries}

In this section, we recall the classical model of (combinational) circuits
over arbitrary semirings, and introduce one simple but useful concept: the set of (exponent) vectors ``produced'' (purely syntactically) by a circuit.

Recall that a (commutative) \emph{semiring} $(R,\suma,\daug)$ consists
of a set $R$ closed under two associative and commutative binary
operations ``addition'' $x\suma y$ and ``multiplication'' $x\daug y$,
where ``multiplication'' distributes over ``addition:'' $x\daug(y\suma
z)=(x\daug y)\suma (x\daug z)$.  That is, in a semiring, we can
``add'' and ``multiply,'' but neither ``subtraction'' nor ``division''
are necessarily possible.  We will assume that semirings contain a
multiplicative neutral element $\vienas\in R$ such that $x\daug
\vienas=\vienas\daug x=x$.

A \emph{circuit} $F$ over a semiring $(R,\suma,\daug)$ is a directed
acyclic graph; parallel edges joining the same pair of nodes are
allowed.  Each indegree-zero node (an \emph{input} node) holds either
one of the variables $x_1,\ldots,x_n$ or a semiring element
$\const{}\in R$.
A circuit is \emph{constant-free} if only variables  $x_1,\ldots,x_n$ are used as inputs.  Every other node, a \emph{gate}, has indegree two
and performs one of the two semiring operations $\suma$ or $\daug$ on
the values computed at the two gates entering this gate.  The
\emph{size} of a circuit is the total number of gates in it.  A
circuit $F$ \emph{computes} a function $f:R^n\to R$ if $F(x)=f(x)$
holds for all $x\in R^n$.

In this article, we will consider circuits over the following three
semirings $(R,\suma,\daug)$: the arithmetic semiring
$(\pRR,+,\times)$, where $\pRR$ is the set of nonnegative real
numbers, the tropical semiring $(\pRR,\min,+)$, and the Boolean
semiring $(\{0,1\},\lor,\land)$. That is, we will consider the
following three types of circuits\footnote{An exception is
  \cref{sec:multilin}, where we also consider non-monotone Boolean
  $(\lor,\land,\neg)$ circuits.}:
\begin{itemize}
\item[$\circ$] $x\suma y:=x+y$ and $x\daug y:=xy$ (monotone arithmetic
  circuits);
\item[$\circ$] $x\suma y:=x\lor y$ and $x\daug y:=x\land y$ (monotone
  Boolean circuits);
\item[$\circ$] $x\suma y:=\min\{x,y\}$ and $x\daug y:=x+y$ (tropical
  circuits).
\end{itemize}
Note that, also over the tropical semiring, ``multiplication''
$(\daug)$ distributes over ``addition'' $(\suma)$ because
$x+\min\{y,z\}=\min\{x+y, x+z\}$.  Also, note that the multiplicative
neutral element $\vienas$ is constant $1$ in arithmetic and Boolean
semirings, but is constant $0$ in the tropical semiring (because
$x+0=x$).

\paragraph{Produced polynomials}
Every circuit $F(x_1,\ldots,x_n)$
over a semiring $(R,\suma,\daug)$ not only computes some function
$f:R^n\to R$, but also \emph{produces} (purely syntactically) an
$n$-variate polynomial over this semiring in a natural way.  Namely,
at each source node holding a constant $\const{}\in R$, the
constant polynomial $P=\const{}$ is produced, and at a source node
holding a variable $x_i$, the polynomial $P=x_i$ is produced.
At an ``addition'' $(\suma)$ gate, the ``sum'' $P\suma Q$
of the polynomials $P$ and $Q$ produced at its inputs is
produced. Finally, the polynomial produced at a
``multiplication'' $(\daug)$ gate is obtained from $P\daug Q$
by the distributivity of $\daug$ over $\suma$; that is, we
``multiply'' $(\daug)$ every monomial of $P$ with every monomial of
$Q$, and take the ``sum'' $(\suma)$ of the obtained monomials. No
terms are canceled along the way. The polynomial produced by the
entire circuit $F$ is the polynomial
$P(x)=\sum_{b\in B}\const{b} \prod_{i=1}^n x_i^{b_i}$
produced at the output gate of $F$; here, $B\subseteq
\NN^n$ is the  set of \emph{exponent vectors} of the polynomial $P$,
 and $x_i^{b_i}$ stands for
the $b_i$-times ``multiplication'' $x_i\daug x_i\daug\cdots\daug x_i$,
and $x_i^0=\vienas$ (the multiplicative neutral element). Since by our
assumption, the underlying semiring contains the multiplicative neutral element~$\vienas$, coefficients $\const{b}$ are semiring
elements\footnote{Because then, by distributivity, we have $x\suma
  x=(\vienas\daug x)\suma(\vienas\daug x)=(\vienas\suma\vienas)x$,
  where $\vienas\suma\vienas$ \emph{is} a semiring element. An example, where
  this is not the case is the semiring $(R,+,\times)$ with
  $R\subset\NN$ being the set of all even integers: then the coefficient
  ``$3$'' of $x+x+x=3x$ is not a semiring element.}.

\begin{rem}
Let us stress the difference between
what circuits \emph{compute} (as functions) and what they
actually \emph{produce} (as formal expressions). The point is that, unlike
when computing polynomial \emph{functions},
no terms are canceled when
  \emph{producing} polynomials. Thus, the
  polynomial function $f$ \emph{computed} by a circuit $F$ may be different from the produced polynomial $P$ due to apparent cancellations of some terms of $P$: unlike for the polynomial $P$ (which only depends on the circuit $F$ itself),
   the function $f$ computed by $F$ already depends on the underlying semiring. For example, over the arithmetic semiring $(\RR,+,\times)$,
   the polynomial produced by the circuit $F=(x+y)(x-y)$
  is $P=x^2+c_1xy+c_2y^2$ with the set $B=\{(2,0), (1,1),(0,2)\}$ of exponent vectors, and coefficients $c_1=1-1=0$ and $c_2=-1$, while the polynomial function \emph{computed} by the circuit $F$ is $f=x^2-y^2$ whose set of exponent vectors is
  $A=\{(2,0),(0,2)\}$.
Note, however, that in monotone Boolean $(\lor,\land)$ and tropical $(\min,+)$ circuits we have no
 cancellations like arithmetic $x-x=0$ because there are no analogs of arithmetic subtraction in the corresponding semirings.
In these circuits, we only have cancellations via absorption
 $x\lor xy=x$ or $\min\{x,x+y\}=x$ when going from
 the \emph{produced} polynomials $P$ to the functions $f$ actually \emph{computed} by the circuits.
 In monotone arithmetic $(+,\times)$ circuits we have no cancellations at all: then $f=P$ (as formal expressions, see \cref{fact2} in \cref{sec:arithmetic}).
  \qed
\end{rem}

\paragraph{Produced sets of exponent vectors}
Of interest for us will be not as much the polynomials
$P(x)=\sum_{b\in B}\const{b} \prod_{i=1}^n x_i^{b_i}$ produced
by circuits $F$ themselves but rather the sets
$B=\Exp{F}\subseteq\NN^n$ of exponent vectors of these
polynomials. These sets can be inductively obtained as follows, where
$\vec{0}$ is the all-$0$ vector, and $\vec{e}_i\in\{0,1\}^n$ has
exactly one $1$ in the $i$th position:

\begin{itemize}
\item[$\circ$] if $F=c\in R$ (a constant), then
 $\Exp{F}=\{\vec{0}\}$;

\item[$\circ$] if $F=x_i$ (input variable), then
 $\Exp{F}=\{\vec{e}_i\}$;

\item[$\circ$] if $F=G\suma H$, then $\Exp{F}=\Exp{G}\cup\Exp{H}$
  (set-theoretic union);

\item[$\circ$] if $F=G\daug H$, then
  $\Exp{F}=\Exp{G}+\Exp{H}:=\{a+b\colon a\in \Exp{G}, b\in \Exp{H}\}$
  (Minkowski sum).
\end{itemize}
Since the set of exponent vectors
of a ``product'' $(\daug)$ of two polynomials is the Minkowski sum of
their sets of exponent vectors, and since no cancellations are performed when producing polynomials,
the set $\Exp{F}\subseteq\NN^n$ of vectors produced by a circuit $F$
is the set of
exponent vectors of the polynomial produced by the
circuit~$F$.

It is clear that the same
circuit with ``addition'' $(\suma$) and ``multiplication'' $(\daug)$
gates can compute \emph{different} functions over different semirings.
Say, the circuit $F = (x\daug y)\suma z$ computes $xy+z$ over the
arithmetic $(+,\times)$ semiring, but computes $\min\{x+y,z\}$ over
the tropical $(\min,+)$ semiring, and computes the Boolean function
$xy\lor z$ over the Boolean $(\lor,\land)$ semiring. It is, however,
important to note that:

\begin{itemize}
\item[$\circ$] The set $\Exp{F}\subseteq\NN^n$ of exponent vectors of the polynomial
  \emph{produced} by a circuit $F$ over any semiring is always the
  same---it only depends on the circuit $F$ itself, not on the underlying
  semiring.
\end{itemize}
This simple observation turned out to be useful when comparing the powers
of circuits over \emph{different} semirings.

\begin{notation}
  We will use standard terminology and notation regarding Boolean
  functions (see, for example, the books by Wegener~\cite{wegener} or
  Crama and Hammer~\cite{hammer}). In particular, for two Boolean
  functions $f,g:\{0,1\}^n\to\{0,1\}$, we write $g\leq f$ iff
  $g(a)\leq f(a)$ holds for all $a\in\{0,1\}^n$.  A Boolean
  \emph{term} is an AND
  of a nonempty set of \emph{literals}, each being a variable $x_i$ or
  its negation~$\bar{x}_i$.  A  term is a \emph{zero term} if it
  contains a variable and its negation.  An \emph{implicant} of a
  Boolean function $f(x_1,\ldots,x_n)$ is a nonzero term $t$ such that
  $t\leq f$ holds, that is, $t(a)=1$ implies $f(a)=1$.  An implicant
  $t$ of $f$ is a \emph{prime implicant} of $f$ if no proper subterm
  $t'$ of $t$ has this property, that is, if $t\leq t'\leq f$ implies
  $t'=t$.  For example, if $f=xy\lor x\bar{y}z$, then $xy$,
  $x\bar{y}z$ and $xz$ are implicants of $f$, but $x\bar{y}z$ is not a
  prime implicant (since $x\bar{y}z\leq xz\leq f$). A Boolean function
  $f$ is \emph{monotone} if $a\leq b$ implies $f(a)\leq f(b)$.
  It is
  well known and easy to show (see, for example,
  \cite[Theorem~1.21]{hammer}) that prime implicants of monotone
  Boolean functions do not contain negated variables.
\end{notation}

\section{Read-k Circuits}
\label{sec:readk}

A monotone Boolean circuit is a circuit over the Boolean semiring
$(R,\suma,\daug)$ with $x\suma y:=x\lor y$ and $x\daug y:= x\land y$;
the domain is $R=\{0,1\}$. To avoid trivialities, we will only
consider monotone Boolean circuits computing \emph{non-constant} Boolean functions $f$. In every $(\lor,\land)$ circuit computing such a function $f$, constant
inputs $0$ and $1$ can be easily eliminated without increasing the
circuit size by iteratively applying $1\land x=x$, $0\land x=0$,
$1\lor x=1$ and $0\lor x=x$. Thus:

\begin{itemize}
\item[$\circ$]  We will always (implicitly) assume that monotone Boolean $(\lor,\land)$ circuits are \emph{constant-free}, that is, have no constants $0$ or $1$ as inputs.
\end{itemize}
Let us fix some notation. The \emph{support} of a vector $a\in\NN^n$ is the set
\[
\supp{a}:=\{i\in[n]\colon a_i\neq 0\}
\]
of its nonzero positions.
For a set $A\subseteq\NN^n$ of vectors, let
\[
\Supp{A}:=\{\supp{a}\colon a\in A\}\subseteq 2^{[n]}
\]
denote the family of supports of vectors of~$A$.
The \emph{upward closure} of a set  $A\subseteq\NN^n$ of vectors
is the set
\[
\Up{A}:=\{b\in\NN^n\colon \mbox{$b\geq a$ for some $a\in A$}\}
\]
of vectors  containing at least one vector of~$A$;
a vector $b$ \emph{contains} a vector $a$  if
$b_i\geq a_i$ holds for all positions $i\in\{1,\ldots,n\}$.

Now let $f:\{0,1\}^n\to\{0,1\}$ be a monotone Boolean function
 A \emph{lowest one} of $f$ is a vector $a\in
\{0,1\}^n$ such that $f(a)=1$ but $f(b)=0$ for every vector $b\leq a$,
$b\neq a$.  We will denote the set of all lowest ones of
$f$ by $\Low{f}$.  Note that the set $\Low{f}$ is always an
\emph{antichain}: $a,b\in \Low{f}$ and $b\leq a$ imply $a=b$.
Since the function $f$ is monotone, for every input vector $x\in\{0,1\}^n$, we have
\[
\mbox{$f(x)=1$ iff $x\geq a$ for some $a\in\Low{f}$ iff $x\in \Up{(\Low{f})}$.}
\]
It is
also easy to see that $a\in f^{-1}(1)$ iff the term
$t_a=\bigwedge_{i\in\supp{a}}x_i$ is an implicant of~$f$.  Thus, $a\in\Low{f}$ iff the term $t_a$ is a \emph{prime}
implicant of~$f$. It is, therefore, useful to keep in mind that if we
view implicants of $f$ as sets of their variables, then
\[
\Low{f} = \mbox{ set of characteristic $0$-$1$ vectors of prime
  implicants of $f$.}
\]
For example, if $f=xz\lor y$, then $\Low{f}=\{(1,0,1), (0,1,0)\}$.
Let us stress that the only reason why  we use such a ``vector-representation'' of prime
implicants is to \emph{unify} (and simplify) the forthcoming comparisons of the powers of \emph{different} types of circuits: Boolean, arithmetic and tropical.

Our starting point is the following simple structural property of sets
of exponent vectors produced by monotone Boolean circuits.

\begin{lem}[Folklore]\label{lem:struct-bool}
  Let $f:\{0,1\}^n\to\{0,1\}$ be a monotone Boolean function,
   $F$ be a monotone Boolean $(\lor,\land)$ circuit, and
  $B_F\subseteq\NN^n$ be the set of exponent vectors produced
  by $F$. Then the following two assertions are equivalent.
  \begin{itemize}
    \Item{i} The circuit $F$ computes $f$.

    \Item{ii} Inclusions $\Supp{\Low{f}}\subseteq\Supp{B_F}$ and
    $B_F\subseteq\Up{(\Low{f})}$ hold.
  \end{itemize}
\end{lem}

\begin{proof}
  Let $A:=\Low{f}$ and $B:=B_F$.  Our Boolean function $f$ is of the form
  $f(x)=\bigvee_{a\in A}\bigwedge_{i\in\supp{a}}x_i$, and the
  Boolean function computed by the circuit $F$ is of the form
  $F(x)=\bigvee_{b\in B}~\bigwedge_{i\in \supp{b}}x_i$.  The
  (ii)~$\Rightarrow$~(i) implication directly follows from a simple
  observation: for every input $x\in\{0,1\}^n$, we have $f(x)=1$ iff
  $\supp{x}\supseteq \supp{a}$ holds for some $a\in A$. Hence, $B\subseteq
  \Up{A}$ yields $F(x)\leq f(x)$, while $\Supp{A}\subseteq\Supp{B}$
  yields $f(x)\leq F(x)$.

  To show the (i)~$\Rightarrow$~(ii) implication, assume that the
  circuit $F$ computes $f$. If $b\not\in\Up{A}$ held for some vector
  $b\in B$, then we would have $\supp{a}\setminus\supp{b}\neq\emptyset$ for all $a\in A$.
  But then, on the input
   $x\in\{0,1\}^n$ with $x_i=1$ iff
  $i\in\supp{b}$, we would have
   $f(x)=0$ while $F(x)=1$, a contradiction. To show the
  inclusion $\Supp{A}\subseteq\Supp{B}$, suppose for the sake of contradiction
  that there is a vector $a\in A$ such that $\supp{b}\neq\supp{a}$
  holds for all vectors $b\in B$.
   Since $B\subseteq \Up{A}$, $a'\leq b$ holds for some vector $a'\in A$.
   So, the proper inclusion  $\supp{b}\subset\supp{a}$ cannot hold, for otherwise, we would have $\supp{a'}\subseteq\supp{b}\subset\supp{a}$ and, hence, also $a'\leq a$  and $a'\neq a$, a contradiction with the set $A$
   being an antichain.
    So, we have
  $\supp{b}\setminus\supp{a}\neq\emptyset$ for all vectors $b\in
  B$. But then $F(a)=0$ while $f(a)=1$, a contradiction.
\end{proof}

A \emph{shadow} of a vector $x\in\RR^n$ is a vector $y\in\RR^n$ with
$\supp{x}=\supp{y}$ (with the same set of nonzero positions as
$x$). The inclusion $\Supp{\Low{f}}\subseteq\Supp{B_F}$ in
\cref{lem:struct-bool} means that every lowest one of
$f$ has at least one shadow in $B_F$.  In general, these
shadows may have large entries, even exponential in the circuit size. In
read-$k$ circuits, we restrict the magnitude of entries of shadows.
Say that a vector $b\in\NN^n$ is $k$-\emph{bounded} if $b_i\leq k$
holds for all $i\in\supp{b}$.

\begin{dfn}[Read-$k$ circuits]\label{def:readk}
  \rm Let $F$ be a monotone $(\lor,\land)$ circuit,
and $B_F\subseteq\NN^n$ be the set of exponent vectors produced by
$F$. The circuit $F$ is a \emph{syntactically read-$k$ circuit} if
 all vectors of $B_F$ are $k$-bounded, and is
a (semantically) \emph{read-$k$ circuit} if every lowest one
of the Boolean function $f:\{0,1\}^n\to\{0,1\}$ computed by $F$ has at least one $k$-bounded shadow in~$B_F$. In particular, $F$ is a read-$1$ circuit
  iff the inclusions $\Low{f}\subseteq B_F\subseteq \Up{(\Low{f})}$
  hold.
\end{dfn}
Thus, $F$ is a (semantically) read-$k$ circuit if
$\Supp{\Low{f}}\subseteq\Supp{B_F\cap\{0,1,\ldots,k\}^n}$ holds, and
is a syntactically read-$k$ circuit if also $B_F\subseteq
 \{0,1,\ldots,k\}^n$ holds.
It is clear that every syntactically read-$k$ circuit is also a (semantically) read-$k$ circuit, but not the vice versa.
Intuitively, a monotone circuit $F$ computing a
monotone Boolean function $f$ is a read-$k$ circuit\footnote{The term
  ``read-$k$ circuit'' is by analogy with the well-known term
  ``read-$k$ times branching program;'' see the discussion at the end
  of \cref{sec:conclusions} (after \cref{probl:dual}).} if it can distinguish each vector
$a\in\Low{f}\subseteq f^{-1}(1)$ from all vectors in $f^{-1}(0)$ by
``reading/accessing'' each $1$-entry of the vector $a$ at most $k$ times.

For a monotone Boolean function $f$, let
\begin{align*}
  \Bool{f}{k} := &\mbox{ min size of a monotone read-$k$
    $(\lor,\land)$ circuit computing $f$.}
\end{align*}

\begin{rem}
  Note that already read-$1$ circuits are ``universal:'' every
  monotone Boolean function $f$ \emph{can} be computed by a read-$1$
  circuit, and even, by a \emph{syntactically} read-$1$ circuit, for example, as an OR of all prime implicants of $f$. But
  read-$k$ circuits for small $k$ can be very inefficient: we will
  show in \cref{sec:gaps} that already the gap
  $\Bool{f}{1}/\Bool{f}{2}$ can be exponential.  \qed\end{rem}

\subsection{Viewing Boolean Circuits as Arithmetic Circuits}

As we already mentioned in \cref{sec:results}, an equivalent and,
apparently, more intuitive definition of read-$k$ $(\lor,\land)$
circuits comes by looking at their arithmetic versions. An arithmetic
circuit is \emph{constat-free} if it has no constants as inputs.  The
\emph{arithmetic $(+,\times)$ version} of a monotone Boolean
$(\lor,\land)$ circuit $F$ is the constant-free $(+,\times)$ circuit
obtained by replacing $\lor$-gates with $+$-gates, and $\land$-gates
with $\times$-gates. That is, we replace the ``addition'' gates by
``addition'' gates, and ``multiplication'' gates by ``multiplication''
gates of the corresponding semirings.
A \emph{shadow monomial} of a Boolean term $\bigwedge_{i\in S}x_i$  is a monomial $\prod_{i\in S}x_i^{d_i}$ with $d_i\geq 1$ for all $i\in S$.

The \emph{formal polynomial} of
a monotone Boolean $(\lor,\land)$ circuit $F$ is the polynomial
$P_F(x)=\sum_{b\in B_F}\const{b}\prod_{i=1}^nx_i^{b_i}$ produced by the
arithmetic $(+,\times)$ version of $F$.  \Cref{lem:struct-bool} implies that the circuit
$F$ computes a (monotone) Boolean function
$f:\{0,1\}^n\to\{0,1\}$ iff the formal polynomial $P_F$ of $F$ has the
following two properties.
\begin{itemize}
\item[(i)] Absorption $x\lor xy=x$: for every monomial
  $\prod_{i=1}^nx_i^{b_i}$ of $P_F$, the Boolean term $\bigwedge_{i\in
    \supp{b}}x_i$ is an implicant of $f$; this is the property
  $B_F\subseteq\Up{(\Low{f})}$ in \cref{lem:struct-bool}.

\item[(ii)] Idempotence $x\land x=x$: every prime implicant of $f$ has at least one shadow
  monomial in $P_F$; this is the property $\Supp{\Low{f}}\subseteq\Supp{B_F}$
  in \cref{lem:struct-bool}.
\end{itemize}
In read-$k$ circuits, the degree of variables in shadow monomials
guaranteed by (ii) is restricted The \emph{individual degree} of a monomial is the maximum degree of its variable.

\begin{dfn}[Arithmetic equivalent of \cref{def:readk}]
  \rm A monotone $(\lor,\land)$ circuit $F$ is a \emph{read-$k$}
  circuit if every prime implicant of $f$ has at least one shadow
  monomial in $P_F$ of individual degree  $\leq k$. If
  the individual degree of \emph{every} monomial of $P_F$ is $\leq k$, then
  $F$ is a \emph{syntactically read-$k$} circuit.
\end{dfn}

\begin{rem}
If $G$ and $H$ are monotone $(\lor,\land)$ circuits, then the formal polynomial of the circuit $F=G\lor H$ is $P_{F}=P_{G}+P_{H}$, and
the formal polynomial of the circuit $F=G\land H$ is $P_{F}=P_{G}\cdot P_{H}$. If $g,h:\{0,1\}^n\to\{0,1\}$ are the (monotone) Boolean functions
computed by the circuits $G$ and $H$, then
every prime implicant of $g\lor h$ is a prime implicant of $g$ or of $h$, and every prime implicant of $g\land h$ is the AND of some  prime implicant of $g$ and some  prime implicant of $h$.
Thus, the OR of two read-$k$  circuits is again a read-$k$ circuit, while the AND of such circuits is a read-$r$ circuit for some $k\leq r\leq 2k$.
In particular,
the ``read parameter'' $k$ does not increase at OR gates: it can only increase at AND gates.
\qed
\end{rem}

\subsection{Reducing the total degree}
\label{sec:strassen}

If $F$ is a \emph{syntactically} read-$k$ $(\lor,\land)$ circuit computing a monotone Boolean function $f:\{0,1\}^n\to\{0,1\}$, then all monomials
of the formal polynomial $P_F$ of $F$
have degree\footnote{As customary, the
  \emph{degree} of a monomial $\prod_{i=1}^n x_i^{b_i}$ is the
  sum $b_1+\cdots+b_n$ of the degrees of its variables. Note that, if a monomial has individual degree $\leq k$, then its degree is $\leq kn$. The \emph{degree} $\deg{P}$ of a polynomial is the maximum
  degree of its monomial.} $\leq kn$. However, if $F$ is a (not necessarily syntactically) read-$k$ circuit, then only shadow monomials in the formal
polynomial $P_F$ must  have degree $\leq kn$: in this case, the polynomial $P_F$ may also have other ``redundant'' monomials of very
large degree, up to $2^rn$ where $r$ is the maximum number of AND
gates along an input-output path in the circuit $F$; about such monomials of $P_F$ we only know property (i) above.

 Still, using a simple observation,
usually attributed to Strassen~\cite{strassen73} (see, for example,
\cite[Theorem~2.2]{Shpilka2010}), one can show that the degree of \emph{all}
monomials in $P_F$ can be decreased till~$kn$.  For a polynomial $P$ of
degree $d$, let $\Hom{P}{i}$ be the sub-polynomial of $P$ consisting of
monomials of degree exactly $i$; hence, $P=\Hom{P}{0}+\Hom{P}{1} +\cdots+\Hom{P}{d}$.

\begin{hlem}[Strassen~\cite{strassen73}]
  If a polynomial $P$ of degree $d$ can be produced by an arithmetic circuit of size
  $s$, then for every $r\leq d$ all homogeneous parts
  $\Hom{P}{0},\Hom{P}{1},\ldots,\Hom{P}{r}$ of $P$ can be simultaneously produced by an
  arithmetic circuit of size $\bigO(sr^2)$.
\end{hlem}

\begin{proof}[Proof sketch]
  The idea is very simple. Take an arbitrary $i\in\{0,1,\ldots,r\}$.  If $P=Q+R$, then $\Hom{P}{i}=\Hom{Q}{i}+\Hom{R}{i}$, and if
  $P=Q\cdot R$, then $\Hom{P}{i}=\sum_{j=0}^i \Hom{Q}{j}\cdot
  \Hom{R}{i-j}$. So, we can take $r+1$
  copies $v_0,v_1,\ldots,v_r$ of each gate $v$ and connect them
  accordingly so that at $v_i$ the homogeneous sub-polynomial of total
  degree $i$ of the polynomial produced at the gate $v$ is produced.
  Every addition  $(+)$ gate is replaced by $r+1$ addition gates, and each multiplication $(\times)$ gate is replaced by $\sum_{i=0}^r(2i+1)=(r+1)^2$ gates. Thus, the obtained circuit has at most $s(r+1)^2$ gates.
\end{proof}

The following easy consequence of this lemma shows that, at
the cost of a relatively small increase in circuit size, we can assume
that formal polynomials of read-$k$ circuits have degree $\leq kn$.

\begin{lem}\label{lem:strassen-trop}
  Let $f(x_1,\ldots,x_n)$ be a monotone Boolean function, and $m$ be the
  maximal number of variables in a prime
  implicant of $f$. If $f$ can be computed by a read-$k$ circuit $F$ of size $s$, then $f$ can also be computed by a read-$k$ circuit $H$ of size at most $s$ times $\bigO(k^2m^2)$ whose formal polynomial $P_H$ has degree $\deg{P_H}\leq km$.
\end{lem}

\begin{proof}
Let $A:=\Low{f}$ be the set of the lowest ones of $f$, that is, of
  characteristic $0$-$1$ vectors of prime implicants of~$f$.
  Since, by our assumption, no prime implicant of $f$ has more than $m$ variables, we have $|\supp{a}|\leq m$ for all $a\in A$.
  Suppose that the function $f$ can be computed by a read-$k$ $(\lor,\land)$ circuit $F$ of size $s$, and let $P(x)=\sum_{b\in B}\const{b}\prod_{i=1}^nx_i^{b_i}$
  be the polynomial produced by the arithmetic $(+,\times)$ version of the   circuit $F$; hence, $P=P_F$ is the formal polynomial of the circuit $F$.
  Since the circuit $F$ computes the function $f$,
  \cref{lem:struct-bool} gives us the inclusions
  $\Supp{A}\subseteq\Supp{B}$ and $B\subseteq\Up{A}$.

  Consider
  the sub-polynomial $P'(x):=\sum_{b\in
    B'}\const{b}\prod_{i=1}^nx_i^{b_i}$ of $P$ whose set of exponent vectors is $B'=\{b\in
  B\colon b_1+\cdots+b_n\leq km\}$.  That is, the polynomial
  $P'=\Hom{P}{0}+\Hom{P}{1} +\cdots+\Hom{P}{km}$
  consists of all terms of $P$ of degree $\leq km$; hence,  $\deg{P'}\leq km$.  Since
  the circuit $F$ is a read-$k$ circuit, for every lowest one
  $a\in\Low{f}$ there is a monomial $\prod_{i=1}^nx_i^{b_i}$ in
  $P$ (a shadow of $a$) with $\supp{b}=\supp{a}$ and $b_i\leq k$ for
  all $i\in\supp{b}$. Since then $b_1+\cdots+b_n\leq k|\supp{a}|\leq
  km$, the polynomial $P'$ contains shadows of \emph{all} prime
  implicants of $f$. This is a crucial property that gives us the inclusion $\Supp{A}\subseteq\Supp{B'}$.

  Since the polynomial $P$ can be produced by a $(+,\times)$ circuit
  of size $s$, Strassen's homogenization lemma implies that the sub-polynomial  $P'$
  of $P$ can be produced by an arithmetic $(+,\times)$ circuit $F'$ of
  size $\bigO(sk^2m^2)$.
  Let $H$ be the Boolean  $(\lor,\land)$ version
  of the circuit $F'$ obtained by replacing $+$ gates by  $\lor$-gates, and $\times$ gates by $\land$-gates.
That is, we replace the ``addition'' gates by
``addition'' gates, and ``multiplication'' gates by ``multiplication''
gates of the corresponding semirings.
The polynomial $P'$ is the formal polynomial $P_H$ of the circuit $H$.
In particular, $\deg{P_H}=\deg{P'}\leq km$.
The set of  ``exponent'' vectors produced by the circuit $H$ is the set $B'$ of exponent  vectors produced by the circuit $F'$.
Hence, the Boolean function computed by $H$ is
the Boolean version $h(x)=\bigvee_{b\in B'}\bigwedge_{i\in
\supp{b}}x_i$ of the polynomial $P'(x)=\sum_{b\in
    B'}\const{b}\prod_{i=1}^nx_i^{b_i}$ produced by $F'$.
Since $\Supp{A}\subseteq\Supp{B'}$ and $B'\subseteq
B\subseteq\Up{A}$, \cref{lem:struct-bool} implies that the
\emph{function} $h$ is the same as our function $f$.
Since the $(\lor,\land)$ circuit $F$ was a read-$k$ circuit, the circuit
$H$ is also a read-$k$ circuit.
\end{proof}

\section{From Monotone Arithmetic to Boolean Read-1}
\label{sec:arithmetic}

A monotone arithmetic $(+,\times)$ circuit is a circuit over the
arithmetic semiring $(R,\suma,\daug)$ with $x\suma y:=x+y$ (arithmetic
addition) and $x\daug y:= x\times y$ (arithmetic multiplication); the
domain is the set $R=\pRR$ of all \emph{nonnegative} real numbers;
hence, the adjective ``monotone:'' since there are no negative
constant inputs (like $-1$), there are no cancellations $x-x=0$.  The
main difference of monotone arithmetic $(+,\times)$ circuits from
Boolean and tropical circuits is that they ``produce what they compute.''

This can be easily shown using the following extension to multivariate
polynomials of a basic fact that no univariate polynomial of degree
$d$ can have more than $d$ roots (see, for example, Alon and
Tarsi~\cite[Lemma~2.1]{tarsi}): If $P$ is a nonzero $n$-variate
polynomial in which every variable has degree $\leq d$, and if
$S\subseteq\RR$ is a set of $|S|\geq d+1$ numbers, then $P(x)\neq 0$
holds for at least one point $x\in S^n$.  This is proved
in~\cite{tarsi} by an easy induction on the number $n$ of variables,
an gives the following property of monotone arithmetic circuits.

\begin{fact}\label{fact2}
  If a monotone arithmetic circuit \emph{computes} a given polynomial,
  then the circuit also \emph{produces} that polynomial.
\end{fact}

\begin{proof}
  Let $F$ be a monotone arithmetic $(+,\times)$ circuit computing a polynomial
  $P$, and let $Q$ be the polynomial produced by $F$. Since the
  circuit $F$ is monotone, it has no negative constants as inputs. So,
  the coefficients in both polynomials $P$ and $Q$ are
  positive. Suppose for the sake of contradiction that the polynomials $P$ and
  $Q$ do not coincide (as formal expressions). Then $P-Q$ is a
  nonzero polynomial of a (possibly large but) \emph{finite} degree
  $d$. By taking any set $S\subseteq \NN$ of $|S|\geq d+1$ numbers,
  the aforementioned result of Alon and Tarsi implies that
  $P(x)-Q(x)\neq 0$ for some $x\in S^n$, a contradiction with our
  assumption that the circuit $F$ computes the polynomial~$P$.
\end{proof}

Say that a polynomial $Q$ is \emph{similar} to a polynomial
$P_A(x)=\sum_{a\in A}\prod_{i=1}^n x_i^{a_i}$ if it is of the form
$Q(x)=\sum_{a\in A}\const{a} \prod_{i=1}^n x_i^{a_i}$ for some
positive integer coefficients $\const{a}\geq 1$ (note that the set $A$ of exponent vectors of $Q$ is the \emph{same} as that of $P_A$).
In particular, the polynomial $P_A$ is similar to itself.
Recall that an arithmetic circuit is \emph{constant-free} if it has no constants as inputs.
For a set $A\subseteq\NN^n$ of vectors, let
\begin{align*}
  \arithm{A} := &\mbox{ min size of a monotone  arithmetic  constant-free $(+,\times)$ circuit producing}\\
  & \mbox{ a polynomial similar to $P_A(x)=\sum_{a\in A}\prod_{i=1}^n
    x_i^{a_i}$.}
\end{align*}
In particular, $\arithm{A}$ is a lower bound on the size of any monotone
arithmetic $(+,\times)$ circuit computing the polynomial $P_A(x)=\sum_{a\in A}\prod_{i=1}^n x_i^{a_i}$.

The model of monotone arithmetic $(+,\times)$ circuits has been studied
in many papers, including
~\cite{gashkov,GS12,jerrum,juk-SIDMA,Raz2011,
schnorr,shamir1980,tiwari1994,valiant80}.
Strong, even exponential lower bounds on $\arithm{A}$ for explicit
sets $A\subseteq\{0,1\}^n$  are  known. Already in 1976, Schnorr~\cite{schnorr} has proved that $\arithm{A}\geq |A|-1$ holds
for every set $A\subseteq\NN^n$ which is \emph{cover-free} in
that $a+b\geq c$ with $a,b,c\in A$ implies $c\in\{a,b\}$.
For example, any set $A\subseteq\{0,1\}^n$ of vectors with $m$ ones, no two of which share $\lfloor m/2\rfloor$ ones in common, is cover-free. Also, as shown in~\cite{schnorr}, the set of characteristic $0$-$1$ vectors of cliques in $K_n$ (viewed as sets of their edges) on the same number of vertices is also cover-free.

Gashkov and Sergeev~\cite{GS12} substantially extended Schnorr's  result by showing that
$\arithm{A}\geq |A|/\max\{t^3,s^2\}-1$ holds for every set
$A\subseteq \NN^n$ which is $(t,s)$-\emph{thin} in the following sense: for
every sets $X,Y\subseteq \NN^n$ of vectors, the inclusion $X+Y\subseteq A$ implies
$|X|\leq t$ or $|Y|\leq s$.
It is easy to see that cover-free sets are
$(1,1)$-thin sets\footnote{Suppose that a set $A$ is \emph{not} a $(1,1)$-thin set.
Then the inclusion $\{x,x'\}+\{y,y'\}\subseteq A$ holds for some vectors $x\neq x'$ and $y\neq y'$. But then the sum $(x+y)+(x'+y')$ of two vectors of $A$ contains a third vector $x+y'$ of $A$, meaning that $A$ is \emph{not} a cover-free set.}. Together with the construction of $(t,t!)$-thin sets
 $A\subseteq\{0,1\}^n$ by Koll\'ar, R\'onyai
and Szab\'o~\cite{KRS96} (via so-called \emph{norm-graphs}), this yields the best known
lower bound
$\arithm{A}\geq 2^{n/2-o(n)}$ on the size of monotone arithmetic
circuits computing any polynomial similar to the explicit multilinear polynomial
$P(x)=\sum_{a\in A}\prod_{i=1}^n x_i^{a_i}$
(see~\cite[Theorem~3]{GS12} or \cite[Appendix~E]{juk-SIDMA} for more
details).

On the other hand, monotone  Boolean read-$k$ $(\lor,\land)$ circuits and monotone arithmetic $(+,\times)$ circuits are interrelated.
By \cref{fact2}, monotone $(+,\times)$ circuits produce what they compute.
Hence, \cref{lem:struct-bool} implies that for every monotone Boolean
function $f:\{0,1\}^n\to\{0,1\}$ and for any $k\geq 1$, we have
\begin{equation}\label{eq:reak-arithm}
\Bool{f}{k} = \min\left\{\arithm{B}\colon \mbox{
    $\Supp{\Low{f}}\subseteq\Supp{B\cap\{0,1,\ldots,k\}^n}$ and
    $B\subseteq\Up{(\Low{f})}$}\right\}\,.
\end{equation}
In particular, for every monotone Boolean
function $f:\{0,1\}^n\to\{0,1\}$ and for every integer $k\geq 1$ there is a
set $B\subseteq\NN^n$ of vectors satisfying the conditions in 
\cref{eq:reak-arithm} such that $\Bool{f}{k}\geq \arithm{B}$ holds, that is,
$\Bool{f}{k}$ is at least the minimum size
of a monotone arithmetic constant-free $(+,\times)$ circuit computing a
polynomial similar to $P(x)=\sum_{b\in B}\prod_{i=1}^n x_i^{b_i}$.
Thus, at least in principle, lower bounds on the size of Boolean read-$k$
$(\lor,\land)$ circuits \emph{can} be obtained by proving lower bounds
on the size of monotone arithmetic $(+,\times)$ circuits. The problem,
however, is that we know only little about the structure of the sets
$B\subseteq\NN^n$ of exponent vectors of the polynomial to be
considered: we only know that the two inclusions in \cref{eq:reak-arithm} hold.  Fortunately, for $k=1$ (read-once circuits), the situation is much better: it
is then enough to prove that $\arithm{B}$ is large for the so-called
``lower envelope'' $B\subseteq\Low{f}$ of the (known to us)
set~$\Low{f}$.

The \emph{lower envelope} $\lenv{B}\subseteq B$ of a set
$B\subseteq\NN^n$ of vectors consists of all vectors $b\in B$ of the
smallest degree, where the \emph{degree} of a vector $b\in\NN^n$ is
the sum $|b|:=b_1+\cdots+b_n$ of its entries. Note that for a monotone
Boolean function $f$, $\lenv{\Low{f}}$ is the set of characteristic
$0$-$1$ vectors of \emph{shortest} implicants of $f$ (those with the smallest
number of variables).  A set $B\subseteq\NN^n$ is \emph{homogeneous}
of degree $m$, if all its vectors have the same degree~$m$; note that
then $\lenv{B}=B$ holds.

If $d$ is the minimum degree of a vector in $B$, then Strassen's
homogenization lemma (see \cref{sec:strassen}) implies that
$\arithm{\lenv{B}}$ is at most $\arithm{B}$ times
$\bigO(d^2)$. However, 
this additional factor $\bigO(d^2)$ can be eliminated
using the fact that the set $B$ contains no vectors of degree \emph{smaller} than~$d$.

\begin{envlem}[Jerrum and Snir~\cite{jerrum}]
  For every $B\subseteq\NN^n$, $\arithm{\lenv{B}}\leq \arithm{B}$.
\end{envlem}

\begin{proof}
  Take a monotone arithmetic constant-free $(+,\times)$ circuit $F$ of size $s=\arithm{B}$
  computing some polynomial $Q(x)=\sum_{b\in B}\const{b}\prod_{i=1}^n  x_i^{b_i}$ similar to the polynomial  $P(x)=\sum_{b\in B}\prod_{i=1}^n x_i^{b_i}$. By \cref{fact2}, the circuit $F$ also \emph{produces} the 
  polynomial $Q$. The polynomial $Q'(x)=\sum_{b\in
  \lenv{B}}\const{b}\prod_{i=1}^n x_i^{b_i}$ is similar to
  the polynomial $P'(x)=\sum_{b\in
  \lenv{B}}\prod_{i=1}^n x_i^{b_i}$. Hence, it is enough to show that
  the polynomial $Q'$ can also be produced by a monotone arithmetic constant-free $(+,\times)$ circuit $F'$ of size at most $s$.
  
 We will obtain the desired circuit $F'$ from the circuit $F$ 
by appropriately discarding some of the edges
  entering $+$-gates.  For a gate $v$ in the circuit $F$, let
  $B_v\subseteq\NN^n$ be the set of exponent vectors of the polynomial
  produced at~$v$, and let $\mdeg{v}$ denote the \emph{minimum} degree $b_1+\cdots+b_n$ of a vector $b\in B_v$. Note that $B_v=B$ holds for the
  output gate $v$ of~$F$. 

  If $v=u\times w$ is a multiplication gate, then
  $B_v=B_u+B_w$ (Minkowski sum). Since the degree of a sum of two
  vectors is the sum of their degrees, we have
  $\mdeg{v}=\mdeg{u}+\mdeg{w}$ and, hence, also
  $\lenv{B_v}=\lenv{B_u+B_w}=\lenv{B_u}+\lenv{B_w}$. So, we do nothing
  in this case. If $v=u+w$ is an addition gate, then $B_v=B_u\cup
  B_w$.  If $\mdeg{u}=\mdeg{w}$, then $\lenv{B_v}=\lenv{B_u\cup
    B_w}=\lenv{B_u}\cup \lenv{B_w}$, and we also do nothing in this case.
  However, if $\mdeg{u}< \mdeg{w}$, then $\lenv{B_v}=\lenv{B_u\cup
    B_w}=\lenv{B_u}$. In this case, we discard the edge $(w,v)$:
  delete the edge $(w,v)$, delete the $+$ operation labeling the gate
  $v$, and contract the remaining edge $(u,v)$. If
  $\mdeg{u}>\mdeg{w}$, then we discard the edge~$(u,v)$.
\end{proof}

A monotone Boolean function $f$ is \emph{homogeneous} if the set
$\Low{f}\subseteq f^{-1}(1)$ of its lowest ones is homogeneous
(meaning that all prime implicants of $f$ have the same number of
variables); note that then $\lenv{\Low{f}}=\Low{f}$ holds.
For a monotone Boolean function $f$, let
\begin{align*}
  \sBool{f}{k} := &\mbox{ min size of a monotone \emph{syntactically} read-$k$  $(\lor,\land)$ circuit computing $f$.}
\end{align*}
It is clear that $\Bool{f}{k}\leq \sBool{f}{k}$ always holds.

\begin{thm}\label{thm:envel}
  For every monotone Boolean function $f$, we have
  \[
  \arithm{\lenv{\Low{f}}}\leq \Bool{f}{1}\leq \sBool{f}{1}\leq \arithm{\Low{f}}\,.
  \]
   In particular, if $f$ is homogeneous, then
  $\arithm{\Low{f}}=\Bool{f}{1}=\sBool{f}{1}$.
\end{thm}

\begin{proof}
  Let $A:=\Low{f}\subseteq f^{-1}(1)$ be the set of the lowest ones of
  $f$; hence, $f(x)=\bigvee_{a\in A}\bigwedge_{i\in\supp{a}}x_i$.  To
  show the first inequality $\arithm{\lenv{A}}\leq \Bool{f}{1}$, let
  $F$ be a monotone read-$1$ Boolean $(\lor,\land)$ circuit of size
  $s=\Bool{f}{1}$ computing $f$, and let $B\subseteq\NN^n$ be the set
  of ``exponent'' vectors produced by $F$. Consider the arithmetic
  $(+,\times)$ version $F'$ of $F$ obtained by replacing $\lor$-gates
  with $+$ gates, and $\land$-gates with $\times$ gates.
  The arithmetic circuit $F'$ has the same number $s=\Bool{f}{1}$ of gates.
  Since the Boolean circuit $F$ is constant-free (by our assumption throughout the paper), its arithmetic version $F'$ is also constant-free.
  So, the circuit $F'$ produces a polynomial $P(x)=\sum_{b\in B}\const{b}\prod_{i=1}^n x_i^{b_i}$ with the same set $B$ of exponent vectors, and some integer coefficients $\const{b}\geq 1$. Since the polynomial $P$ is similar to the polynomial $\sum_{b\in B}\prod_{i=1}^n x_i^{b_i}$, we have $\arithm{B}\leq s$. By  Envelope Lemma,  we have $\arithm{\lenv{B}}\leq \arithm{B}\leq s$. So, it is enough to show that
  $\lenv{A}=\lenv{B}$.

Since the (Boolean) circuit $F$
  is a read-$1$ circuit, we know that the inclusions $A\subseteq
  B\subseteq \Up{A}$ hold.
Let $m$ be the minimum degree $|a|=a_1+\cdots+a_n$ of a vector $a\in
  A$; hence, $\lenv{A}=\{a\in A\colon |a|=m\}$.  The inclusion
  $B\subseteq \Up{A}$ means that for every vector $b\in B$ there is a
  vector $a\in A$ such that $b\geq a$. Together with the inclusion
  $A\subseteq B$, this implies that $\lenv{B}=\{b\in B\colon |b|=m\}$
  and $\lenv{A}\subseteq\lenv{B}$. To show the inclusion
  $\lenv{B}\subseteq\lenv{A}$, take an arbitrary vector $b\in \lenv{B}$; hence, $|b|=m$. Since $B\subseteq \Up{A}$, there must be a vector $a\in A$ such that $b\geq a$. Since the set $A$ has no
  vectors of degree $< m=|b|$, we have  $b=a\in\lenv{A}$.
Thus, $\lenv{A}=\lenv{B}$ holds, as desired.

To show the inequality $\sBool{f}{1}\leq \arithm{A}$, let $F$ be a constant-free arithmetic $(+,\times)$ circuit of size
$s=\arithm{A}$ computing some polynomial $Q$ similar to $P(x)=\sum_{a\in
  A}\prod_{i=1}^n x_i^{a_i}$. Thus, the set of exponent vectors of the polynomial $Q$ is the same set $A=\Low{f}$ as that of the polynomial $P$.
 Let $F'$ be the Boolean
$(\lor,\land)$ version of $F$ obtained by replacing every $+$-gate
with a $\lor$-gate, and every $\times$-gate with a $\land$-gate. The
resulting Boolean circuit $F'$ produces the same set $A$ of exponent
vectors as the arithmetic circuit $F$. Hence, $F'$ computes the
Boolean version $f(x)=\bigvee_{a\in A}\bigwedge_{i\in\supp{a}}x_i$ of
the polynomial~$P$. Since the set $A=\Low{f}$ of exponent vectors of
$P$ consists of only $0$-$1$ vectors, the circuit $F'$ is a
\emph{syntactically} read-$1$ circuit. Hence, $\sBool{f}{1}\leq s$, as desired.
\end{proof}

\subsection{An easy lower bound for monotone arithmetic circuits}
\label{sec:explicit}
The goal of this section is to demonstrate that strong lower bounds on the size of monotone
arithmetic $(+,\times)$ circuits and, hence (by \cref{thm:envel}),
also on the size of monotone Boolean read-$1$ $(\lor,\land)$ circuits can be proved fairly
easily.

The weakness of monotone arithmetic circuits lies in \cref{fact2}:
unlike for monotone \emph{Boolean} circuits (where both idempotence $x\land x=x$ and
absorption $x\lor xy=x$ are allowed), monotone \emph{arithmetic} circuits ``produce what
they compute.'' This weakness results in the following ``balanced
decomposition property'' for arithmetic circuits computing
homogeneous\footnote{A polynomial $f(x)=\sum_{a\in
    A}\const{a}\prod_{i=1}^nx_i^{a_i}$ is monotone if $\const{a}>0$
  for all $a\in A$, and is \emph{homogeneous} of degree $\deg{f}=m$ if
  $a_1+\cdots+a_n=m$ holds for all $a\in A$.} polynomials observed
already by Hyafil~\cite[Theorem~1]{hyafil} and
Valiant~\cite[Lemma~3]{valiant80}.

\begin{decomplem}
If a homogeneous polynomial $f$ of degree $m\geq 3$
  can be computed by a monotone arithmetic $(+,\times)$ circuit of
  size $s$, then $f$ can be written as a sum $f=g_1h_1+\cdots+g_t h_t$
  of $t\leq s$ products of homogeneous polynomials such that $m/3\leq
  \deg{g_i}\leq 2m/3$ for all $i=1,\ldots,t$.
\end{decomplem}

In particular, the inclusions $\mon{g_ih_i}\subseteq\mon{f}$ hold for
all $i=1,\ldots,t$, where $\mon{f}$ denotes the set of all monomials
of~$f$. That is every monomial of $g_ih_i$ must also be a monomial of the 
computed polynomial (no ``redundant'' monomials can be produced); 
this is in stark contrast with Boolean or tropical circuits, where
nothing similar holds.
 The proof of Decomposition Lemma is simple. If a circuit computes the
polynomial $f$, then (by \cref{fact2}) it also produces that
polynomial. Since $f$ is homogeneous, polynomials $g_v$ produced at
intermediate gates $v$ are also homogeneous. By walking backward from
the output gate, and by always choosing that of the two input gates
$v$ with larger $\deg{g_v}$, we will find a gate $v$ with $m/3\leq
\deg{g_v}\leq 2m/3$.  Hence, the polynomial $f$ is of the form
$f=g_vh_v + \cdots$ for some polynomial $h_v$.  Replace the gate $v$
by constant $0$, and argue by induction on circuit size.

\begin{ex}[Perfect matchings]\label{ex:matchings}
  The \emph{perfect matching} function is a monotone Boolean function
  $\match{n}$ of $n^2$ variables, one for each edge of $K_{n,n}$, such
  that $\match{n}(x)=1$ iff the subgraph of $K_{n,n}$ specified by the input vector $x\in\{0,1\}^{n\times n}$
  contains a perfect matching. We will use the
  decomposition lemma to  prove the
  following lower bound:
  \[
  \mbox{For $f=\match{n}$, we have $\Bool{f}{1}=\arithm{\Low{f}}=
    2^{\Omega(n)}$.}
  \]
  \begin{proof}
    The set $\Low{f}$ of the lowest ones of this function consists of
    $|\Low{f}|=n!$ characteristic $0$-$1$ vectors $a\in\{0,1\}^{n\times n}$ of all perfect
    matchings (viewed as sets of their edges).  Since the set
    $\Low{f}$ is homogeneous (of degree $n$), \Cref{thm:envel} yields
    $\Bool{f}{1}=\arithm{\Low{f}}$, and it remains to prove the lower
    bound $\arithm{\Low{f}}=2^{\Omega(n)}$.
The polynomial $\sum_{a\in\Low{f}}\prod_{i,j=1}^nx_{i,j}$
is the well known permanent polynomial $\perm{n}(x)=\sum_{\pi}\prod_{i=1}^nx_{i,\pi(i)}$, 
where the sum is over all $n!$ permutations $\pi:[n]\to[n]$.
Hence, $\arithm{\Low{f}}$ is the minimum size
$s$ of a monotone arithmetic circuit $F$ computing a polynomial
similar to the polynomial $\perm{n}$. 

To apply
the decomposition lemma,
take an arbitrary polynomial of the form $gh$ with
$\mon{gh}\subseteq\mon{\perm{n}}$ and $\deg{g}=r$ for some $n/3\leq r\leq 2n/3$. Since
the polynomial $\perm{n}$ is homogeneous of degree $n$, the polynomials
$g$ and $h$ must be homogeneous of degrees $n/3\leq\deg{g}=r\leq
2n/3$ and $n/3\leq \deg{h}=n-r\leq 2n/3$.  Fix any two monomials
$p\in\mon{g}$ and $q\in\mon{h}$; hence, $p$ corresponds to a
matching in $K_{n,n}$ with $r$ edges, and $q$ corresponds to a
matching in $K_{n,n}$ with $n-r$ edges; since the polynomial $\perm{n}$
is multilinear, these two matchings must be vertex-disjoint.  A
matching in $K_{n,n}$ with $r$ edges is contained in only $(n-r)!$
perfect matchings. So, $|\mon{h}|=|\mon{p\cdot h}|\leq (n-r)!$ and
$|\mon{g}|=|\mon{g\cdot q}|\leq r!$. This gives an upper bound
$|\mon{gh}|=|\mon{g}|\cdot |\mon{h}|\leq r!(n-r)!$ on the number
of monomials in the polynomial $gh$. By Decomposition Lemma, the circuit $F$ must have $s\geq
|\mon{f}|/|\mon{gh}|\geq n!/r!(n-r)!=\binom{n}{r}$ gates. Since
$\binom{n}{r}\geq \binom{n}{n/3}$ for every $n/3\leq r\leq 2n/3$,
we have $\arithm{\Low{f}}\geq \binom{n}{n/3}$.
\end{proof}
\end{ex}

\begin{rem}\label{rem:JS}
  By using the permanent equivalent of Laplace's expansion rule for
  determinants, Jerrum and Snir~\cite{jerrum} have shown that the
  permanent polynomial $\perm{n}$ (of $n^2$ variables) can be
  computed by a monotone arithmetic $(+,\times)$ circuit using at most
  $t:=n(2^{n-1}-1)$ multiplication $(\times)$ gates (using a more subtle argument as in the proof above, it is proved in ~\cite{jerrum} that
   this number of multiplication gates is also
  necessarily.) On the other hand, an argument similar to that used by
  Alon and Boppana~\cite[Lemma 3.15]{alon-boppana87} for monotone Boolean circuits implies that
if a multilinear $n$-variate polynomial $P$ can be computed by a monotone arithmetic $(+,\times)$ circuit with $t$ multiplication $(\times)$ gates,
then a polynomial similar to $P$ can be computed by a monotone arithmetic $(+,\times)$ circuit with $t$ multiplication $(\times)$ gates and $\bigO(tn+t^2)$ addition $(+)$ gates. Thus,
  the minimum size of a monotone arithmetic $(+,\times)$ circuit
  computing $\perm{n}$ is $2^{\Theta(n)}$.
  \qed
\end{rem}

\section{From Tropical (min,+) to Boolean Read-$1$}
\label{sec:tropical}

A tropical $(\min,+)$ circuit is a circuit over the tropical semiring
$(R,\suma,\daug)$ with $x\suma y:=\min\{x,y\}$ and $x\daug y:=x+y$
(arithmetic addition); the domain is the set $R=\pRR$ of all
nonnegative real numbers.

Note that in the tropical $(\min,+)$ semiring, powering $x_i^{a_i}=x_i\daug
x_i\daug\cdots\daug x_i$ ($a_i\in\NN$ times) turns into multiplication
by scalars $a_ix_i=x_i+x_i+\cdots +x_i$.  So, a (generic) monomial
$\prod_{i=1}^n x_i^{a_i}$ turns into the tropical ``monomial''
$\skal{a,x}=a_1x_1+\cdots+a_nx_n$, the scalar product of vectors $a$
and $x$, and a polynomial $\sum_{a\in A}\const{a}\prod_{i=1}^n
x_i^{a_i}$ turns into the \emph{tropical} $(\min,+)$ polynomial
$f(x)=\min_{a\in A} \skal{a,x}+\const{a}$ with ``exponent'' vectors
$a\in A$ and ``coefficients'' $\const{a}\in\RR_+$;
vectors $a\in A$ are usually called \emph{feasible solutions}
of the corresponding minimization problem.
For example, an arithmetic
polynomial $P(x,y)=2x^3y+4xy^2$ turns into the tropical polynomial
$f(x,y)=\min\{3x+y+2,x+2y+4\}$

A $(\min,+)$ circuit $F$ \emph{approximates} a given minimization problem $f:\pRR^n\to\pRR$  within a factor $k\geq 1$
if for every input weighting $x\in\pRR^n$, the inequalities $f(x)\leq
F(x)\leq k\cdot f(x)$ hold. That is, the circuit is not allowed to output any better (smaller) than optimal values but is allowed to output
up to $k$ times worse than the optimal values.
 In
particular, the circuit $F$ \emph{solves} the problem $f$ exactly
(approximates $f$ within factor $k=1$) if $F(x)=f(x)$ holds for all
$x\in\pRR^n$.

The minimization problem $f:\pRR^n\to\pRR$ represented by a tropical
polynomial $f(x)=\min_{a\in A} \skal{a,x}+\const{a}$ is
\emph{constant-free} if $\const{a}=0$ for all $a\in A$.  Combinatorial
optimization problems are usually constant-free.  For example, in the
famous \emph{MST problem} (minimum weight spanning tree problem) on a
given graph $G$, the goal is to compute the constant-free $(\min,+)$
polynomial $f(x)=\min_{a\in A} \skal{a,x}$, where $A$ is the set of
characteristic $0$-$1$ vectors of spanning trees of $G$ (viewed as
sets of their edges). In the not less prominent \emph{assignment
  problem}, $A$ is the set of characteristic $0$-$1$ vectors of
perfect matchings, etc.

Since constant inputs in $(\min,+)$ circuits can only affect the ``coefficients'' $\const{a}$, such inputs ``should'' be of little use when solving constant-free minimization problems. This intuition was confirmed in \cite[Lemma~3.2]{JS20a} using a simple argument (which we include for completeness):  when dealing with tropical $(\min,+)$ circuits approximating constant-free minimization problems, we can safely restrict ourselves to constant-free circuits.  The
\emph{constant-free version} of a $(\min,+)$ circuit $F$ is obtained
by replacing all constant inputs with constant~$0$.

\begin{lem}\label{lem:const-free}
  If a $(\min,+)$ circuit $F$ approximates a constant-free
  minimization problem  within a factor $k\geq 1$, then the constant-free version of $F$ also
  approximates this problem within the same factor.
\end{lem}

\begin{proof}
  Let $f(x)=\min_{a\in A} \skal{a,x}$ be a (constant-free)
  minimization problem approximated by the circuit $F$ within factor
  $k$, and let $g(x)=\min_{b\in B}\ \skal{x,b}+\const{b}$ be the
  tropical $(\min,+)$ polynomial produced by~$F$. Since constant
  inputs can only affect the ``coefficients'' $\const{b}$, the
  polynomial produced by the constant-free version $\cf{F}$ of~$F$ is
  the constant-free version $\cf{g}(x)=\min_{b\in B}\ \skal{x,b}$ of
  the polynomial $g(x)$.  Since the circuit $F$ approximates $f$
   within the factor $k\geq 1$,
  the inequalities $f(x)\leq g(x)\leq k\cdot f(x)$ hold for all
  weightings $x\in\pRR^n$.  We have to show that $f(x)\leq
  \cf{g}(x)\leq k\cdot f(x)$ also holds for all $x\in\pRR^n$.  Since
  the constants $\const{b}$ are nonnegative, we have
  $\cf{g}(x)\leq g(x)\leq k\cdot f(x)$ for all $x\in\pRR^n$.  To show
  that $f(x)\leq \cf{g}(x)$ holds as well, suppose for the sake of
  contradiction that $f(x_0)>\cf{g}(x_0)$ holds for some input
  weighting $x_0\in\pRR^n$. Then the difference $d=f(x_0)-\cf{g}(x_0)$
  is positive.  We can assume that the constant $c:=\max_{b\in
    B}\const{b}$ is also positive, for otherwise, there would be
  nothing to prove. So the constant $\lambda:=2c/d$ is positive.
  Since $\cf{g}(x_0)=f(x_0)-d$, we obtain $g(\lambda x_0)\leq
  \cf{g}(\lambda x_0)+c =\lambda\cdot \cf{g}(x_0)+c
  =\lambda[f(x_0)-d]+c=\lambda \cdot f(x_0)-c=f(\lambda x_0)-c$, which
  is strictly smaller than $f(\lambda x_0)$, a contradiction with
  $f(x)\leq g(x)$ for all $x\in\pRR^n$.
\end{proof}

If a tropical $(\min,+)$ circuit $F$ solves
the minimization problem $f_A(x)=\min_{a\in A}\ \skal{a,x}$ on a given set $A$ of feasible solutions, then the set $B_F\subseteq\NN^n$ of ``exponent'' vectors does not need to coincide with $A$. For example,
the circuit $F=\min\{x,y\}+\min\{x,y\}$ solves the minimization problem
 $f_A=\min\{2x,2y\}$ on $A=\{(2,0),(0,2)\}$ by producing the set
 $B_F=\{(2,0), (1,1), (0,2)\}$ of ``exponent'' vectors. Note that the vector $(1,1)$ is the convex combination $\tfrac{1}{2}(2,0)+\tfrac{1}{2}(0,1)$ of the vectors of $A$.

And indeed,
 using a version of Farkas' lemma, Jerrum and Snir~\cite{jerrum}
have shown that the structures of the sets $A$ and $B_F$ are related via convexity.
Namely, a $(\min,+)$ circuit $F$ solves the minimization problem $f_A$ of a set $A\subseteq\NN^n$ of feasible solutions if and only if every vector
of $A$ contains\footnote{A vector $x\in\RR^n$ \emph{contains} a vector $y\in\RR^n$ if $x\geq y$ holds, that is, if $x_i\geq y_i$ for all $i=1,\ldots,n$.} some convex combination of vectors in $B_F$,
and every vector of $B_F$ contains some convex combination of
vectors in~$A$. The ``if'' direction here is simple: the scalar product
of $x\in\RR^n$ with a convex combination of some collection of vectors is at least the minimum scalar product
of $x\in\RR^n$ with some of these vectors. The Farkas lemma is used
in~\cite{jerrum} to show the ``only if'' direction.

In the case of $0$-$1$ optimization, feasible solutions $a\in A$ are $0$-$1$ vectors. In this case, the following properties of sets $B_F$ can be easily proved without any use of Farkas' lemma.

\begin{lem}[Structure]\label{lem:struct-trop}
Let $F$ be a constant-free
$(\min,+)$ circuit, and $B_F\subseteq\NN^n$ be the set of ``exponent''
vectors produced by~$F$.
  If the circuit $F$ approximates the minimization problem
  $f_A(x)=\min_{a\in A}\skal{a,x}$ on an antichain
  $A\subseteq\{0,1\}^n$ within a factor $k\geq 1$, then $B_F\subseteq\Up{A}$ and for every vector
  $a\in A$ there is a vector $b\in B_F$ such that $\supp{b}=\supp{a}$
  and $\skal{a,b}\leq k\cdot\skal{a,a}$.  In particular, if $k=1$, then
  $A\subseteq B_F\subseteq\Up{A}$.
\end{lem}

\begin{proof}
Let $B:=B_F$.
  Since the circuit $F$ is constant-free, the problem solved by $F$ is
  of the form $f_B(x)=\min_{b\in B}\skal{b,x}$. Since $F$ approximates
  the problem $f_A$ within the factor $k$, inequalities $f_A(x)\leq
  f_B(x)\leq k\cdot f_A(x)$ hold for all $x\in\pRR^n$. To show the
  inclusion $B\subseteq\Up{A}$, take an arbitrary vector $b\in B$, and
  consider the weighting $x\in\{0,1\}^n$ such that $x_i:=0$ for
  $i\in\supp{b}$, and $x_i:=1$ for $i\not\in\supp{b}$. Take a vector
  $a\in A$ on which the minimum $f_A(x)=\skal{a,x}$ is achieved. Then
  $\skal{a,x}=f_A(x)\leq f_B(x)\leq \skal{b,x}=0$. Thus,
  $\supp{a}\subseteq\supp{b}$. Since $b\in\NN^n$ and $a$ is a $0$-$1$
  vector, this yields $a\leq b$, as desired.

  Now take an arbitrary vector $a\in $A, and consider the weighting
  $x\in\{1,kn+1\}^n$ with $x_i:=1$ for all $i\in \supp{a}$ and
  $x_i:=kn+1$ for all $i\not\in\supp{a}$. Let $b\in B$ be a vector on
  which the minimum $f_B(x)=\skal{b,x}$ is achieved. Then
  $\skal{b,x}=f_B(x)\leq k\cdot f_A(x)\leq k\cdot \skal{a,x}=k\cdot
  \skal{a,a}\leq kn$.  If $b_i\neq 0$ held for some
  $i\not\in\supp{a}$, then we would have $\skal{b,x}\geq b_ix_i =
  b_i(kn+1) > kn$, a contradiction. Thus, the inclusion
  $\supp{b}\subseteq\supp{a}$ holds.  Since $B\subseteq\Up{A}$, there
  is a vector $a'\in A$ such that $a'\leq b$. Hence,
  $\supp{a'}\subseteq \supp{b}\subseteq\supp{a}$. Since both $a$ and
  $a'$ are $0$-$1$ vectors, this yields $a'\leq a$ and, since the set
  $A$ is an antichain, we have $a'=a$ and, hence, also
  $\supp{b}=\supp{a}$. By the definition of the weighting $x$, this
  yields
  $\skal{a,b}=\skal{b,x}=f_B(x)\leq k\cdot\skal{a,a}$, as desired.
\end{proof}

For a finite set $A\subseteq\NN^n$ of vectors, let
\begin{align*}
  \kMin{A}{k} := &\mbox{ smallest size of a $(\min,+)$ circuit
    approximating the minimization}\\
  &\mbox{ problem $g(x)=\min_{a\in A}\skal{a,x}$ on $A$ within the factor
    $k$.}
\end{align*}

\begin{thm}\label{thm:trop-read-once}
  Let $f$ be a monotone Boolean function,  $m$ be the largest
  number of variables in a prime implicant of $f$, $k\geq 1$ be an integer, and  $r=(k-1)m+1$. Then
  \[
  \Bool{f}{r}\leq \kMin{\Low{f}}{k}\leq
  \Bool{f}{k}\,.
  \]
   In particular,
  $\kMin{\Low{f}}{1}=\Bool{f}{1}$.
\end{thm}

\begin{proof}
  Let $A:=\Low{f}\subseteq\{0,1\}^n$ be the set of
  the lowest ones of $f$. To show the inequality
  $\kMin{A}{k}\leq\Bool{f}{k}$, let $F$ be a read-$k$ $(\lor,\land)$
  circuit of size $\Bool{f}{k}$ computing $f$, and let
  $B:=B_F\subseteq\NN^n$ be the set of ``exponent'' vectors produced by
  $F$.
  The tropical $(\min,+)$ version $F'$ of $F$ is a constant-free
  $(\min,+)$ circuit obtained by replacing $\lor$  gates with $\min$ gates,
  and $\land$ gates with addition gates.
  That is, we (again) replace ``addition''   gates by
``addition'' gates, and ``multiplication'' gates by ``multiplication''
gates of the corresponding semirings.
  The circuit $F'$ produces the
  same set $B$ of ``exponent'' vectors. Hence, the circuit $F'$ solves
  the minimization problem $g_B(x)=\min_{b\in B}\skal{b,x}$ on the set
  $B$. The minimization problem on the given set $A$ is
  $g_A(x)=\min_{a\in A}\skal{a,x}$.  It thus remains to show that for
  every input weighting $x\in\pRR^n$ the inequalities $g_A(x)\leq
  g_B(x)\leq k\cdot g_A(x)$ hold. So, take an arbitrary input
  weighting $x\in\pRR^n$.

  Since the Boolean circuit $F$ computes $f$, \cref{lem:struct-bool}
  gives us the inclusion $B\subseteq\Up{A}$, that is, for every vector
  $b\in B$ there is a vector $a\in A$ such that $b\geq a$. Since the
  weights are nonnegative, this gives the first inequality $g_A(x)\leq
  g_B(x)$. To show the second inequality $g_B(x)\leq k\cdot g_A(x)$,
  take a vector $a\in A$ on which the minimum $g_A(x)=\skal{a,x}$
  on the input weighting $x\in\pRR^n$ is
  achieved, and let $S:=\supp{a}$ be the support of $a$. Since $F$ is
  a read-$k$ circuit, there is a vector $b\in B$ such that
  $\supp{b}=S$ and $b_i\leq k$ for all $i\in S$.  Thus, $g_B(x)\leq
  \skal{b,x}=\sum_{i\in S}b_ix_i\leq k\cdot \sum_{i\in S}x_i= k\cdot
  \skal{a,x}= k\cdot g_A(x)$, as desired.

  To show the inequality, $\Bool{f}{r}\leq \kMin{A}{k}$ take a
  tropical $(\min,+)$ circuit $F$ of size $\kMin{A}{k}$
  approximating the (constant-free) minimization problem
  $g_A(x)=\min_{a\in A}\skal{a,x}$ within the factor $k$, and let
   $B:=B_F\subseteq\NN^n$ be the set
  of ``exponent'' vectors produced by the circuit $F$.  By
  \cref{lem:const-free}, we can assume that the circuit $F$ is
  constant-free, that is, has no nonzero constant inputs.  So, the
  polynomial $g_B(x)=\min_{b\in B}\skal{b,x}$ produced by $F$ is also
  constant-free.  The Boolean $(\lor,\land)$ version $F'$ of $F$
  (obtained by replacing $\min$ gates with $\lor$ gates, and addition gates
  with $\land$ gates) produces the same set $B$ of ``exponent'' vectors.
  Since the $(\min,+)$ circuit $F$ approximates the problem $g_A$ within the factor $k$,
  \cref{lem:struct-trop} implies that the set $B$ has the following
  two properties: (i) $B\subseteq\Up{A}$, and (ii) for every vector
  $a\in A$ there is a vector $b\in B$ such that $\supp{b}=\supp{a}$
  and $\skal{a,b}\leq k\cdot\skal{a,a}$. This, in particular, yields
  the inclusion $\Supp{A}\subseteq\Supp{B}$. Together with
  $B\subseteq\Up{A}$, \cref{lem:struct-bool} implies that the
  $(\lor,\land)$ circuit $F'$ computes our Boolean function $f$, and it
  remains to show that $F'$ is a read-$r$
  circuit for $r:=(k-1)m+1$.

  To show this, take an arbitrary lowest one $a\in A$ of $f$, and let
  $S:=\supp{a}$ be its support; hence, $|S|\leq m$.  By
  property (ii), there is a vector $b\in B$ such that $\supp{b}=S$ and
  $\sum_{i\in S}b_i\leq k|S|$. It remains to show that the vector $b$ is $r$-bounded.
 Suppose for the sake of contradiction that
  $b_j\geq r+1=(k-1)m+2$ holds for some position $j\in S$. Since $b_i\geq
  1$ holds for all $i\in S$, we then have $\sum_{i\in S}b_i\geq r+1 +
  (|S|-1)=(k-1)m+2+ (|S|-1) \geq (k-1)|S|+2+ (|S|-1)=k|S|+1$, a
  contradiction with $\sum_{i\in S}b_i\leq k|S|$.
\end{proof}

\section{From Non-Monotone Multilinear to Monotone Read-1}
\label{sec:multilin}

Due to the lack of strong lower bounds for  (non-monotone)
arithmetic $(+,\times,-)$ circuits, and because they seem to be the
most intuitive circuits for computing \emph{multilinear} polynomials,
a successful approach has been to consider a restriction called
``multilinearity'' of arithmetic circuits, first introduced by Nisan
and Wigderson~\cite{nisan97}.

Recall that a polynomial is multilinear if it does not have any
variable with degree larger than $1$.  An arithmetic $(+,\times,-)$
circuit $F$ is \emph{syntactically multilinear} if
the two subcircuits rooted at inputs of any multiplication $(\times)$  gate
have no input variables in common.
The circuits $F$ is (semantically)
\emph{multilinear} if the polynomial \emph{functions} computed at its gates polynomials are multilinear. For example, the polynomial
function $f=y$ computed at a gate producing the polynomial
$P=x^2+y-x^2$ is multilinear. Raz~\cite[Proposition~2.1]{Raz-ACM} observed that  minimal semantically multilinear $(+,\times,-)$ \emph{formulas}
(circuits whose underlying graphs are trees) are syntactically
multilinear. It remains not clear if every semantically
multilinear \emph{circuit} can be efficiently simulated by a
syntactically multilinear circuit.

There are several impressing results concerning multilinear (as well
as syntactically multilinear) arithmetic $(+,\times,-)$ circuits and
formulas; see, for example, the surveys~\cite{ChenKW11,Shpilka2010}. In
particular, Raz~\cite{Raz-ACM} proved that any multilinear arithmetic
\emph{formula} computing the permanent or the determinant of an
$n\times n$ matrix is of size $n^{\Omega(\log n)}$.  Furthermore,
Raz~\cite{Raz06} proved that a gap between multilinear arithmetic
\emph{formulas} and \emph{circuits} can be super-polynomial. Proving
super-polynomial lower bounds of the size of multilinear arithmetic
\emph{circuits} remains an open problem.

Due to the lack of even super-linear lower bounds on the size of (unrestricted)
DeMorgan $(\lor,\land,\neg)$ circuits, and by analogy with arithmetic
circuits, the multilinearity restriction was also imposed on DeMorgan
circuits.  Recall that a DeMorgan $(\lor,\land,\neg)$
circuit\footnote{Let us note that every Boolean $(\lor,\land,\neg)$
  circuit (with negations applied to any gates, not necessarily to only
  inputs) can be easily transformed into an equivalent DeMorgan
  $(\lor,\land,\neg)$ circuit by only doubling the circuit size (see,
  e.g., \cite[p.~195]{wegener}): we double all AND and OR gates, one
  output of a pair is negated, the other one not; after that, we can
  move negation gates toward the input variables by applying the
  DeMorgan rules.}  $F$ is an $(\lor,\land)$ circuit whose inputs are
the variables $x_1,\ldots,x_n$ and their negations
$\bar{x}_1,\ldots,\bar{x}_n$.  As before, the \emph{size} of a circuit
is the total number of gates in it. A \emph{monotone} Boolean circuit
is a DeMorgan circuit without negated input literals as inputs.

A DeMorgan $(\lor,\land,\neg)$ circuit $F$ is \emph{syntactically
  multilinear} if the two subcircuits rooted at inputs of any AND gate
have no input literals of the same variable in common. For example,
the circuit $F=(x\lor x\bar{y})y$ is \emph{not} syntactically
multilinear. Sengupta and Venkateswaran~\cite{sengupta} considered
the \emph{connectivity} function which accepts an input
$x\in\{0,1\}^{\binom{n}{2}}$ iff the subgraph $G_x$ of $K_n$  specified by the characteristic $0$-$1$ vector $x$ of its set of edges is connected. By adopting the proof of Jerrum
and Snir~\cite{jerrum} of a lower bound $(4/3)^{n-1}/n$ on the minimum
size of monotone arithmetic $(+,\times)$ circuits computing the directed spanning
tree polynomial, it was shown in~\cite{sengupta} that every monotone
syntactically multilinear $(\lor,\land)$ circuit computing the connectivity
function must have at least $\sqrt{(4/3)^{n-1}/n}$ gates.

Krieger~\cite{Krieger07} has shown that if the set $\Low{f}$ of the lowest
ones of a monotone Boolean function $f$ is cover-free (that is, if
$a,b,c\in\Low{f}$ and $a+b\geq c$ imply $c\in\{a,b\}$), then every
syntactically multilinear DeMorgan $(\lor,\land,\neg)$ circuit computing $f$
must have at least $|\Low{f}|-1$ gates.

\begin{rem}
 As mentioned in \cref{sec:explicit},  already in 1976, Schnorr~\cite{schnorr} has proved a general lower bound
  $\arithm{A}\geq |A|-1$ on the monotone \emph{arithmetic} circuit
  complexity of polynomials, whose sets $A\subseteq\NN^n$ of exponent vector are
  cover-free.  This surprising similarity of  Krieger's bound
  in~\cite{Krieger07} with Schnorr's bound, as well as a possibility
  to adopt in~\cite{sengupta} the argument of Jerrum and Snir, already
  served as an indication that there ``should'' be some \emph{general}
  relation between multilinear DeMorgan $(\lor,\land,\neg)$ circuits and
  monotone arithmetic circuits.  Our
  \cref{thm:multilinear,thm:mon-multilin} below give such a relation,
  even for semantically (not only syntactically) multilinear
  DeMorgan $(\lor,\land,\neg)$ circuits: such circuits are \emph{not} stronger
  than monotone arithmetic circuits.
\qed\end{rem}

Following the analogy with arithmetic circuits, Ponnuswami and
Venkateswaran~\cite{Ponnuswami04} relaxed the \emph{syntactic}
multilinearity restriction of Boolean circuits to their
\emph{semantic} multilinearity. A Boolean function
$f(x_1,\ldots,x_n)$ \emph{depends} on the $i$th variable $x_i$ if
$f(a)\neq f(b)$ holds for some two vectors $a,b\in\{0,1\}^n$ that
differ only in the $i$th position. The following simple fact is well known;
see, for example, \cite[Theorem~1.17]{hammer}.

\begin{fact}[Folklore]\label{fact3}
  A Boolean function $f(x_1,\ldots,x_n)$ depends on the $i$th variable
  $x_i$ iff $x_i$ or $\bar{x}_i$ appears in at least one prime
  implicant of~$f$.
\end{fact}

\begin{proof}
  The ``only if'' direction follows from the obvious fact that every
  Boolean function $f$ is an OR of its prime implicants.  So, if
  neither $x_i$ nor $\bar{x}_i$ appears in any prime implicant of $f$,
  then $f$ does not depend on the $i$th variable.  To show the ``if''
  direction, let $t=zt'$ be a prime implicant of $f$, where
  $z\in\{x_i,\bar{x}_i\}$. Since the implicant $t$ is prime, the term
  $t'$ is not an implicant of $f$. That is, there is a vector
  $a\in\{0,1\}^n$ such that $t'(a)=1$ but $f(a)=0$ and, hence, also
  $t(a)=0$ (because $t\leq f$). Let $b$ be the vector $a$ with its
  $i$th bit $a_i$ replaced by $1-a_i$. Then $t(b)=1$ and, hence, also
  $f(b)=1$, meaning that the function $f$ depends on the $i$th
  variable.
\end{proof}

Say that two Boolean functions are \emph{independent} if they
depend on disjoint sets of variables.

\begin{dfn}[Multilinear circuits]
  \rm A DeMorgan $(\lor,\land,\neg)$ circuit $F$ is \emph{multilinear}
  (or \emph{semantically multilinear}) if the two Boolean functions
  $g$ and $h$ computed at the inputs to any AND gate are independent.
\end{dfn}
By \cref{fact3}, the functions $g$ and $h$  are independent iff their prime implicants share no common variables (negated or not). However, the terms actually \emph{produced} at the gates \emph{computing} these functions
\emph{can} share common variables. This explains the use of the adjective
``semantically.''
For example, the circuit $F=(x\lor xy)(\bar{y}\lor \bar{y}z)$ is not
syntactically multilinear, but is (semantically) multilinear because
$g=x\lor xy$ depends only on $x$, while $h=\bar{y}\lor \bar{y}z$
depends only on~$y$.

The \emph{upward closure} of a Boolean function
$f:\{0,1\}^n\to\{0,1\}$ is the monotone Boolean function
\[
\up{f}(x):=\bigvee_{z\leq x}f(z)\,.
\]
For example, the upward closure of the parity function
$f=x_1\oplus x_2\oplus\cdots\oplus x_n$ is $\up{f}=x_1\lor
x_2\lor\cdots\lor x_n$.  Note that $\up{f}=f$ holds for
\emph{monotone} functions~$f$.
Also note that $\up{(g\lor h)}=\up{g}\lor\up{h}$ holds for any Boolean functions
$g,h:\{0,1\}^n\to\{0,1\}$. Thus, if $\PI{f}$ is the set of all prime implicants of a Boolean function
$f:\{0,1\}^n\to\{0,1\}$, then
\[
\up{f}=\up{\big(\bigvee_{t\in \PI{f}}t\big)}=\bigvee_{t\in \PI{f}}\up{t}=\bigvee_{t\in \PI{f}}\pos{t}\,,
\]
where $\pos{t}$ is the \emph{positive factor} of a term $t$ obtained from $t$ by replacing every negated literal
$\bar{x}_i$ with constant~$1$. Thus, the upward closure $\up{f}$ of any
Boolean function $f$ is the OR of positive factors of prime implicants
of $f$.

A \emph{lowest one} of a (not necessarily monotone) Boolean function
$f:\{0,1\}^n\to\{0,1\}$ is a vector $a\in\{0,1\}^n$ such that $f(a)=1$
but $f(b)=0$ for all $b < a$; for vectors $a,b\in\RR^n$ we write $b<a$
if $b\leq a$ and $b_i < a_i$ for at least one position~$i$.
 Let, as before, $\Low{f}\subseteq
f^{-1}(1)$ denote the set of all lowest ones of~$f$. For example,
the set $\{\vec{e}_1,\ldots,\vec{e}_n\}$ of $n$ unit vectors is the
set of the lowest ones of $x_1\oplus x_2\oplus\cdots\oplus x_n$ as well as
of $x_1\lor x_2\lor\cdots\lor x_n$. Note that, unlike for \emph{monotone} Boolean functions, $a\in\Low{f}$ does not exclude that $f(c)=0$ holds for some vectors $c\geq a$. For example, $a=(1,0)$ is a lowest one of the function $f=x\bar{y}\lor\bar{x}y=x\oplus y$, but $f(c)=0$ for $c=(1,1)$.

\begin{rem}\label{rem:env-of-up}
It is easy to
verify that the lowest ones of a function $f:\{0,1\}^n\to\{0,1\}$ and of its downward closure $\up{f}$ are  the
same, that is, $\Low{\up{f}}=\Low{f}$ holds. Indeed, if $a\in
\Low{f}$, then $f(a)=1$ but $f(b)=0$ for all $b<a$. Hence, also
$\Low{\up{f}}(a)=1$ but $\Low{\up{f}}(b)=0$ for all $b<a$. This shows
the inclusion $\Low{f}\subseteq \Low{\up{f}}$.  If $a\in
\Low{\up{f}}$, then $\Low{\up{f}}(a)=1$ but $\Low{\up{f}}(b)=0$ and,
hence, also $f(b)=0$ holds for all $b<a$. Since $\up{f}(a)=1$ still holds,
this can happen only if
$f(a)=1$. Hence, the converse inclusion $\Low{\up{f}}\subseteq
\Low{f}$ also holds.
\qed
\end{rem}

For a Boolean function $f$, let $\linBool{f}$ denote the minimum size
of a (semantically) multilinear DeMorgan $(\lor,\land,\neg)$ circuit
computing $f$. For a \emph{monotone} Boolean function $f$, let
$\mlinBool{f}$ denote the minimum size of a \emph{monotone}
(semantically) multilinear $(\lor,\land)$ circuit computing~$f$.
It is clear that $\linBool{f}\leq  \mlinBool{f}$ holds for every \emph{monotone} Boolean function~$f$.

\begin{thm}[Arbitrary functions]\label{thm:multilinear}
  For every Boolean function $f$, we have
  \[
  \arithm{\lenv{\Low{f}}}\leq \Bool{\up{f}}{1} \leq
  \mlinBool{\up{f}}\leq \linBool{f}\,.
  \]
\end{thm}

We will actually prove slightly stronger results:
(i) instead of just an inequality $\Bool{\up{f}}{1} \leq
  \mlinBool{\up{f}}$, we shown that every multilinear \emph{monotone} $(\lor,\land)$ circuit \emph{is} a read-$1$ circuit (\cref{lem:multilin-to-read-1}), and (ii)
  instead of just an inequality $\mlinBool{\up{f}}\leq \linBool{f}$, we shown that
  if a multilinear DeMorgan $(\lor,\land,\neg)$ circuit $F$ computes a Boolean function $f$, then the monotone $(\lor,\land)$ circuit $\pos{F}$, obtained from $F$ by replacing every negated input variable $\bar{x}_i$ with constant $1$, is also multilinear and computes~$\up{f}$ (\cref{lem:multilin-pos}).

Let us first prove the following easy consequence of
\cref{thm:multilinear} for multilinear $(\lor,\land,\neg)$ circuits
computing \emph{monotone} Boolean functions $f$, and then prove \cref{thm:multilinear}
itself.

\begin{thm}[Monotone functions]\label{thm:mon-multilin}
  For every monotone Boolean function $f$, we have
  \[
  \arithm{\lenv{\Low{f}}}\leq \Bool{f}{1} \leq \linBool{f} =
  \mlinBool{f} \leq\arithm{\Low{f}}\,.
  \]
  In particular, of $f$ is also homogeneous, then
  $\arithm{\Low{f}}=\Bool{f}{1}=\linBool{f}=\mlinBool{f}$.
\end{thm}

\begin{proof}
  The inequality $\arithm{\lenv{\Low{f}}}\leq\Bool{f}{1}$ is given by
  \cref{thm:envel}.  Since $f$ is monotone, we have $\up{f}=f$. So,
  the equality $\linBool{f}=\mlinBool{f}$ follows from a trivial upper
  bound $\linBool{f}\leq \mlinBool{f}$ and from the lower bound
  $\linBool{f}\geq \mlinBool{\up{f}}=\mlinBool{f}$ given by
  \cref{thm:multilinear}. It therefore remains to prove the upper
  bound $\mlinBool{f}\leq \arithm{\Low{f}}$.

  For this, let $A:=\Low{f}$ be the set of the lowest ones of $f$, and
  take a monotone  arithmetic constant-free $(+,\times)$ circuit $F$ of
  size $s=\arithm{A}$ computing some polynomial $P(x)=\sum_{a\in
    A}\const{a}\prod_{i=1}^nx_i^{a_i}$ similar to $\sum_{a\in
    A}\prod_{i=1}^nx_i^{a_i}$. By \cref{fact2}, the circuit $F$ also
  \emph{produces} the polynomial $P$. Let $F'$ be the Boolean $(\lor,\land)$ version
  of the circuit $F$ obtained by replacing each $+$-gate by a
  $\lor$-gate, and each $\times$-gate by $\land$-gate.  The circuit
  $F'$ produces the same set $A$ of exponent vectors and, hence,
  computes our Boolean function $f(x)=\bigvee_{a\in
    A}\bigwedge_{i\in\supp{a}}x_i$. Since $A$ consists of only $0$-$1$
  vectors, the polynomial $P$ produced by the arithmetic circuit $F$
  is multilinear, meaning that the polynomials produced at inputs of
  any multiplication gate cannot share any variables in common. Thus,
  the Boolean version $F'$ of $F$ is (even syntactically) multilinear.
\end{proof}

\begin{rem}
  Ponnuswami and Venkateswaran~\cite{Ponnuswami04} proved a lower bound
  $\mlinBool{f}=\Omega(2^{.459 n})$ for the perfect matching function
  $f=\match{n}$ (which we considered in \cref{sec:explicit}).  On the
  other hand, using arguments tighter than we used in \cref{ex:matchings},
  Jerrum and Snir~\cite{jerrum} have proved a lower bound
  $\arithm{\Low{f}}\geq n(2^{n-1}-1)$ for $f$.  The function
  $f$ is homogeneous (each prime implicant has $n$
  variables).  So, by \cref{thm:mon-multilin}, the same lower bound  $\linBool{f}\geq n(2^{n-1}-1)$ holds even for non-monotone circuits.
  \qed\end{rem}

\begin{rem}
Lingas~\cite{lingas} has proved a lower bound $\mlinBool{f}\geq \arithm{\Low{f}}/\bigO(m^2)$ for every monotone homogeneous Boolean function, where $m$ is the number of variables in the prime implicants of $f$. On the other hand,
\cref{thm:mon-multilin} shows that, for homogeneous monotone
functions $f$, we actually have the equality
$\mlinBool{f}=\arithm{\Low{f}}$, and even the equality
$\linBool{f}=\arithm{\Low{f}}$. That is, multilinear (not necessarily
monotone)
DeMorgan $(\lor,\land,\neg)$ circuits computing monotone homogeneous Boolean functions have the \emph{same} power as monotone arithmetic constant-free
$(+,\times)$ circuits.
\qed\end{rem}

\subsection{Proof of Theorem~\ref{thm:multilinear}}
Since lowest ones of a Boolean function $f$ and of its upward closure
$g:=\up{f}$ are the same (see \cref{rem:env-of-up}), we have $\lenv{\Low{f}}=\lenv{\Low{g}}$.  By
\cref{thm:envel}, $\arithm{\lenv{\Low{f}}}=\arithm{\lenv{\Low{g}}} \leq \Bool{g}{1}$. This shows the first
inequality in~\cref{thm:multilinear}. To prove the remaining two
inequalities of~\cref{thm:multilinear}, we first establish (in
\cref{lem:properties}) the behavior of sets of lowest ones
as well as of upward closures of functions computed at
the gates of DeMorgan $(\lor,\land,\neg)$ circuits.
When doing this, we will use the following simple property of independent Boolean functions following from \cref{fact3}. For $0$-$1$ vectors $a,b\in\{0,1\}^n$, $a\lor
b\in\{0,1\}^n$ denotes their componentwise~OR. For example, if
$a=(1,1,0)$ and $b=(0,1,1)$ then $a+b=(1,2,1)$ but $a\lor b=(1,1,1)$.

\begin{fact}\label{fact4}
  Let $g,h:\{0,1\}^n\to\{0,1\}$ be Boolean functions, and let
  $b\in\Low{g}$ and $c\in\Low{h}$ be their lowest ones. If $g$ and $h$
  are independent, then $\supp{b}\cap\supp{c}=\emptyset$
  \mbox{\rm (}hence, also $b\lor c=b+c$\mbox{\rm )} and $g(b\lor
  c)=h(b\lor c)=1$.
\end{fact}

\begin{proof}
  Note that $t(b)=1$ holds for some prime implicant $t=\bigwedge_{i\in
    S}x_i\land\bigwedge_{j\in T}\bar{x}_j$ of $g$ with $S=\supp{b}$:
  we have $S\subseteq\supp{b}$ since $t(b)=1$, and $\supp{b}\subseteq
  S$ since $g(b)=1$ and $b$ is a \emph{lowest} one of~$g$. So, since
  $g$ and $h$ are independent, the disjointness
  $\supp{b}\cap\supp{c}=\emptyset$ follows from \cref{fact3}.  In
  particular, $b+c=b\lor c$ is a $0$-$1$ vector.  Since the function
  $g$ does not depend on any variable $x_i$ with $i\in\supp{c}$, we
  have $g(b\lor c)=g(b+c)=g(b+\vec{0})=g(b)=1$. Similarly, since
  function $h$ does not depend on any variable $x_i$ with
  $i\in\supp{b}$, we also have $h(b\lor
  c)=h(b+c)=h(\vec{0}+c)=h(c)=1$.
\end{proof}

Recall that the Minkowski sum of two sets $A,B\subseteq\RR^n$ is the set $A+B=\{a+b\colon a\in A, b\in B\}$.

\begin{lem}\label{lem:properties}
  Let $g,h:\{0,1\}^n\to\{0,1\}$ be Boolean functions.
  \begin{itemize}
    \Item{i} If $f=g\lor h$, then
    $\Low{f}\subseteq\Low{g}\cup\Low{h}$ and
    $\up{f}=\up{g}\lor\up{h}$.

    \Item{ii} If $f=g\land h$ and $g,h$ are independent, then
    $\Low{f}\subseteq\Low{g}+\Low{h}$ and $\up{f}=\up{g}\land\up{h}$.
  \end{itemize}
\end{lem}

\begin{proof}
  To show (i), let $f=g\lor h$. The inclusion $\Low{f}\subseteq\Low{g}\cup\Low{h}$ is trivial:
  if $a\in\Low{f}$, then $g(a)=1$ or $h(a)=1$, and both $g(b)=0$
  and $h(b)=0$ hold for every vector $b<a$. Thus, either $a\in\Low{g}$
  or $a\in\Low{h}$, as desired.   To show the inequality
  $\up{f}\leq \up{g}\lor\up{h}$, take any vector $x\in\{0,1\}^n$ for
  which $\up{f}(x)=1$ holds; hence, $x\geq a$ for some lowest one
  $a\in\Low{f}$. Then, as we have just shown, either $a\in\Low{g}$ or
  $a\in\Low{h}$ (or both) hold. Hence, either $\up{g}(x)=1$ or $\up{h}(x)=1$, as
  desired.  To show the opposite inequality $\up{f}\geq
  \up{g}\lor\up{h}$, take any vector $x\in\{0,1\}^n$ for which
  $\up{g}(x)=1$ holds. Then $g(z)=1$ and, hence, also $f(z)=1$ holds
  for some $z\leq x$, meaning that $\up{f}(x)=1$, as desired. The same
  happens if $\up{h}(x)=1$.

  To show (ii), let $f=g\land h$, where the functions $g$ and $h$ are
  independent. Take an arbitrary lowest one $a\in \Low{f}$. Since then
  $g(a)=1$ and $h(a)=1$, there are lowest ones $b\in\Low{g}$ and
  $c\in\Low{h}$ such that $b\leq a$ and $c\leq a$; hence, $a\geq b\lor
  c$.  Since the functions $g$ and $h$ are independent, \cref{fact4}
  yields $b\lor c=b+c$ and $g(b+c)=h(b+c)=1$; hence, also
  $f(b+c)=1$. Since $a\geq b+c$ and since vector $a$ is a
  \emph{lowest} one of $f$, this yields the equality $a=b+c$; hence,
  $a\in\Low{g}+\Low{h}$.  This shows the inclusion
  $\Low{f}\subseteq\Low{g}+\Low{h}$.

  The inequality $\up{f}\leq \up{g}\land\up{h}$ is trivial and holds
  for arbitrary (not necessarily independent) Boolean functions $g$ and
  $h$: if $\up{f}(x)=1$, then $f(z)=1$ holds for some vector $z\leq x$
  and, hence, both $g(z)=1$ and $h(z)=1$ also hold. To show the converse
  inequality $\up{g}\land\up{h}\leq \up{f}$ for independent functions
  $g$ and $h$, take any vector $x\in\{0,1\}^n$ for which both
  $\up{g}(x)=1$ and $\up{h}(x)=1$ hold.  Then $g(b)=1$ and $h(c)=1$
  hold for some lowest one $b\leq x$ of $g$ and for some lowest one
  $c\leq x$ of $h$. Since the functions $g$ and $h$ are independent,
  \cref{fact4} yields $g(b\lor c)=h(b\lor c)=1$ and, hence, also
  $f(b\lor c)=1$. Since $b\lor c\leq x$, this yields $\up{f}(x)=1$, as
  desired.
\end{proof}

\begin{rem}\label{rem:dependent}
In general, if  $f=g\land h$, and if the functions $g$ and $h$ are not independent, then
even the inclusion $\Low{f}\subseteq\Low{g}\lor\Low{h}:=\{b\lor c\colon b\in\Low{g},x\in\Low{h}\}$ does not need to hold.
    Take, for example,
    $g=x\bar{y}\,\bar{z}\lor xy$ and
  $h=\bar{x}\,\bar{y}z\lor yz$; hence, $f=xyz$. The functions $g$ and $h$ are dependent (their prime implicants
  $xy$ and $yz$ share a common variable $y$).  The only lowest one of $f$ is $a=(1,1,1)$, the
  only lowest one of $g$ is $b=(1,0,0)$, and the only lowest one of
  $h$ is $c=(0,0,1)$. But $a\neq b\lor c$. Also, in general, the
  inequality $\up{g}\land\up{h}\leq \up{f}$ does not need to
  hold. Take, for example,  $g=x\bar{y}\lor z$
  and $h=\bar{x}y\lor z$; hence, $f=z$. On the vector
  $a=(1,1,0)$, we have $\up{f}(a)=f(a)=0$, but $\up{g}(a)\geq
  g(1,0,0)=1$ and $\up{h}(a)\geq h(0,1,0)=1$.
\qed\end{rem}

We now turn to the actual proof of the remaining two inequalities
$ \mlinBool{\up{f}}\leq \linBool{f}$ and $\Bool{\up{f}}{1}\leq \mlinBool{\up{f}}$ claimed in \cref{thm:multilinear}. This is done in
\cref{lem:multilin-pos,lem:multilin-to-read-1} bellow.

We can view every DeMorgan $(\lor,\land,\neg)$ circuit $F(x)$
computing a Boolean function $f(x)$ of $n$ variables as a
\emph{monotone} $(\lor,\land)$ circuit $H(x,y)$ of $2n$ variables with
the property that $f(x)=H(x,\overline{x})$ holds for all
$x\in\{0,1\}^n$, where
$\overline{x}=(\overline{x}_1,\ldots,\overline{x}_n)$ is the
complement of $x=(x_1,\ldots,x_n)$. For example, if $x=(1,0,1,1)$, then
$\bar{x}=(0,1,0,0)$. The \emph{monotone version} of the
circuit $F(x)$ is the monotone circuit $\pos{F}(x)=H(x,\vec{1})$
obtained by replacing every negated input literal $\overline{x}_i$ in the circuit $F$ with constant~$1$.

Take, for example, the circuit
$F=(x\bar{y}\lor z)(\bar{x}y\lor z)$ computing the Boolean function
$f=z$. Its monotone version $\pos{F}=(x\cdot 1\lor z)(1\cdot y\lor z)=(x\lor z)(y\lor z)$ computes the Boolean function  $xy\lor z$ which is  \emph{different} from
the upward closure $\up{f}=z$ of $f$. The following lemma shows that
this cannot happen for \emph{multilinear} circuits~$F$.

\begin{lem}[Multilinear to monotone multilinear]\label{lem:multilin-pos}
  Let $F$ be a DeMorgan $(\lor,\land,\neg)$ circuit computing a
  Boolean function $f$.  If $F$ is multilinear, then the circuit
  $\pos{F}$ is also multilinear and computes $\up{f}$.  In particular,
  $\mlinBool{\up{f}}\leq \linBool{f}$ holds.
\end{lem}

\begin{proof}
  Suppose that the circuit $F(x)=H(x,\bar{x})$ is multilinear.
   To show that the (monotone) circuit $\pos{F}(x)=H(x,\vec{1})$ is multilinear,
  let $g$ and $h$ be the Boolean functions computed at some AND gate of
  the circuit $F(x)=H(x,\bar{x})$. Since the circuit $F(x)$ is
  multilinear, the functions $g$ and $h$ are independent.  By
  \cref{fact3}, this happens precisely when their prime implicants share no
  variables in common. Since the upward closure of any
Boolean function is the OR of positive factors of its prime implicants,
\cref{fact3} implies that the functions  $\up{g}$ and $\up{h}$ are also independent.

Let us now show that the monotone version $\pos{F}=H(x,\vec{1})$ of $F$
  computes the upward closure $\up{f}$ of $f$.  Upward closures of
  input variables $x_i$ are the variables $\up{x_i}=x_i$ themselves,
  while upward closures of negated input variables $\bar{x}_i$ are
  constant-$1$ functions $\up{\bar{x}_i}=1$.  Let $g$ and $h$ be the
  Boolean functions computed at the two inputs of an arbitrary gate of
  $F$. If this is an OR gate, then \cref{lem:properties} yields the
  equality $\up{(g\lor h)}=\up{g}\lor\up{h}$. If this is an AND gate
  then, since the circuit $F$ is multilinear, the functions $g$ and
  $h$ are independent, and \cref{lem:properties} also yields the
  equality $\up{(g\land h)}=\up{g}\land\up{h}$. Thus, in the circuit
  $\pos{F}=H(x,\vec{1})$, the upward closures $\up{g}$ of the
  functions $g$ computed at the gates of $F$ are computed. Since this
  also holds for the output gate of $F$, at which the function $f$ is
  computed, the upward closure $\up{f}$ of $f$ is computed at this
  gate in the circuit $\pos{F}$, as desired.
\end{proof}

\begin{lem}[Monotone multilinear to read-$1$]\label{lem:multilin-to-read-1}
  Monotone multilinear Boolean circuits are read-$1$ circuits.  In
  particular, $\Bool{f}{1}\leq \mlinBool{f}$ holds for every monotone
  Boolean function~$f$.
\end{lem}

\begin{proof}
  Let $F$ be a monotone multilinear $(\lor,\land)$ circuit computing a
  monotone Boolean function $f$.  Let $B_F\subseteq\NN^n$ be the set of
  ``exponent'' vectors produced by $F$. By \cref{lem:struct-bool}, the
  inclusion $B_F\subseteq\Up{(\Low{f})}$ holds. So, to show that $F$ is
  a read-$1$ circuit, we have only to show that also the inclusion
  $\Low{f}\subseteq B_F$ holds, i.e., that every lowest one $a\in
  \Low{f}$ of $f$ is produced by the circuit~$F$.

  Let $G$ and $H$ be the subcircuits of $F$ whose output gates
  enter the output gate of $F$, and let $g$ and $h$ be the
  monotone Boolean functions computed by these subcircuits. Let also
  $B_G\subseteq\NN^n$ and  $B_H\subseteq\NN^n$ be the sets of
  ``exponent'' vectors produced by the subcircuits $G$ and $H$. We
  argue by induction on the number $s$ of gates in~$F$. In the basis
  case $s=1$, we have $G=x_i$ and $H=x_j$ for some
  $i,j\in[n]$. Hence, $B_G=\{\vec{e}_i\}=\Low{g}$ and
  $B_H=\{\vec{e}_j\}=\Low{h}$.  So, if $F=G\lor H=x_i\lor x_j$
  then $B_F=B_G\cup B_H=\{\vec{e}_i,\vec{e}_j\} =\Low{f}$. If
  $F=G\land H=x_i\land x_j$, then $i\neq j$ due to the
  multilinearity of the circuit~$F$. Hence, also in this case, we have
  $B_F=B_G+B_H=\{\vec{e}_i+\vec{e}_j\}=\Low{f}$.

  Now suppose that the lemma holds for all monotone multilinear
  $(\lor,\land)$ circuits of size at most $s$, and let $F$ be a
  monotone multilinear $(\lor,\land)$ circuit of size $s+1$. Since the
  circuit $F$ is multilinear, both subcircuits $G$ and $H$ are also
  multilinear. Since each of $G$ and $H$ has at most $s$ gates,
  the lemma holds for both these subcircuits. Thus, both inclusions
  $\Low{g}\subseteq B_G$ and $\Low{h}\subseteq B_H$ hold.

  If $F=G\lor H$, then $B_F=B_G\cup B_H$ and \cref{lem:properties}
  gives the inclusion $\Low{f}\subseteq \Low{g}\cup \Low{h}$. So,
  the desired inclusion $\Low{f}\subseteq B_F$ follows from the
  induction hypothesis.  If $F=G\land H$, then $B_F=B_G+B_H$
  (Minkowski sum). Since the circuit $F$ is multilinear, the functions
  $g$ and $h$ are independent, and \cref{lem:properties} yields
  $\Low{f}\subseteq \Low{g}+\Low{h}$. So, the desired inclusion
  $\Low{f}\subseteq B_F$ follows again from the induction hypothesis.
\end{proof}

\subsection{Multilinear Circuits Impede Zero Terms}
\label{sec:impede}

\Cref{lem:multilin-pos} rises a natural question: if $F=F(x,\bar{x})$
is a DeMorgan $(\lor,\land,\neg)$ circuit computing a Boolean function
$f$, when does its monotone version $\pos{F}=F(x,\vec{1})$
computes~$\up{f}$? It can be easily shown that such are exactly
DeMorgan circuits $F$ that ``impede'' zero terms in the following
sense.

Every DeMorgan $(\lor,\land,\neg)$ circuit $F$  not only computes some Boolean function $f:\{0,1\}^n\to\{0,1\}$ but also \emph{produces} (purely
syntactically) a unique set $\T{F}$ of Boolean terms in a natural way:

\begin{itemize}
\item[$\circ$] if $F=z$ is an input literal $z\in\{x_i,\bar{x}_i\}$, then $\T{F}=\{z\}$;

\item[$\circ$] if $F=F_1\lor F_2$, then $\T{F_1\lor F_2}=\T{F_1}\cup\T{F_2}$;

\item[$\circ$] if $F=F_1\land F_2$, then
$\T{F}=\left\{t_1\land t_2\colon t_1\in \T{F_1}, t_2\in
  \T{F_2}\right\}$.
\end{itemize}
During the production of terms, the ``annihilation''
law $x\land \bar{x}=0$ is not used.  So, $\T{F}$ can contain
\emph{zero terms}, that is, terms containing a variable $x_i$ together
with its negation~$\bar{x}_i$. For example, the set $\T{F}=\{x\bar{y},x\bar{y}z,xy\bar{y}, xy\bar{y}z\}$
of terms produced by the circuit
$F=(x\lor xy)(\bar{y}\lor\bar{y}z)$ contains two zero terms $xy\bar{y}$
and $xy\bar{y}z$.

Recall that the \emph{positive factor} $\pos{t}$ of a Boolean term $t$ is obtained by
replacing every its negated literal $\bar{x}_i$ with constant~$1$.
Let us say that a DeMorgan $(\lor,\land,\neg)$ circuit $F$ computing a
Boolean function $f$ \emph{impedes zero terms} if positive factors $\pos{t}$ of zero terms $t\in\T{F}$ produced by $F$ (if there are any) are implicants of
$\up{f}$, that is, if $\pos{t}\leq \up{f}$ holds for every zero term
$t\in\T{F}$. Note that such a circuit does \emph{not} forbid a
production of zero terms as such, but rather ``impedes'' produced zero
terms to unfold the full power of cancellations $x\land\bar{x}=0$.

\begin{fact}\label{fact:impede}
  Let $F$ be a DeMorgan $(\lor,\land,\neg)$ circuit computing a
  Boolean function $f$ with $f(\vec{0})=0$. The circuit $\pos{F}$
  computes $\up{f}$ if and only if $F$ impedes zero terms.
\end{fact}

\begin{proof}
  Since $F$ computes $f$, we have $f=\bigvee_{t\in T}t$, where
  $T=\T{F}$ is the set of all terms produced by the circuit~$F$.
  Since $f(\vec{0})=0$, none of the terms $t\in T$ consist of solely
  negated variables.  For every term $t$, we have $\up{t}=0$ (the
  constant $0$ function) if $t$ is a zero term, and $\up{t}=\pos{t}$
  if $t$ is a nonzero term.  So, if $T_0\subseteq T$ is the set of all
  zero terms in $T$, then (where the second equality follows from
  \cref{lem:properties}(i)):
  \[
  \up{f}=\up{\bigg(\bigvee_{t\in T}t\bigg)}=\bigvee_{t\in T}\up{t}
  =\bigvee_{t\in T\setminus T_0}\up{t}
  =\bigvee_{t\in T\setminus T_0}\pos{t}\leq \bigvee_{t\in T}\pos{t}=\pos{F}
  \]
  with the equality iff $\pos{t}\leq \up{f}$ holds for all terms $t\in
  T\setminus T'$.
\end{proof}

DeMorgan $(\lor,\land,\neg)$ circuits that do not produce zero terms
at all obviously impede zero terms.  Such circuits were considered by
several authors, starting with Kuznetsov~\cite{kuznetsov} (already in
1981, under the name ``circuits without null-chains''), where he
 proved a surprisingly large
lower bound $2^{n/3}$ on the size of such circuits computing an
explicit $n$-variate Boolean function. Sengupta and
Venkateswaran~\cite{sengupta-noncancel} also considered DeMorgan
$(\lor,\land,\neg)$ circuits that do not produce zero terms (under the
name of ``non-cancellative circuits''). They showed that for every
such circuit $F$ computing a Boolean function $f$, the monotone
version $\pos{F}$ of $F$ computes $\up{f}$. Since non-cancellative
circuits produce no zero terms, this also follows from
\cref{fact:impede}.

Multilinear DeMorgan $(\lor,\land,\neg)$ circuits already \emph{can}
produce zero terms. For example, the DeMorgan $(\lor,\land,\neg)$
circuit $F=(x\lor xy)(\bar{y}\lor\bar{y}z)$ computing $f=x\bar{y}$ is
multilinear but produces zero terms $xy\bar{y}$ and $xy\bar{y}z$.
Still, together with \cref{lem:multilin-pos}, \cref{fact:impede}
implies that multilinear DeMorgan circuits impede the produced zero
terms as well.

\section{The Read-1/Read-2 Gap Can be Exponential}
\label{sec:gaps}

\Cref{thm:trop-read-once,thm:envel,thm:multilinear,thm:mon-multilin}
show that read-$1$ $(\lor,\land)$ circuits are
not weaker than monotone arithmetic constant-free  $(+,\times)$
circuits, not weaker than tropical $(\min,+)$ circuits,
 and not weaker than (non-monotone) multilinear $(\lor,\land,\neg)$
circuits.
 Let us now show that already read-$2$ $(\lor,\land)$
circuits can be much smaller than read-$1$ $(\lor,\land)$ circuits. For this, let $n=m^2$ and consider the
following monotone Boolean function of $n$ variables whose inputs are Boolean $m\times
m$ matrices $x=(x_{i,j})$:
\[
\mbox{$\lines{n,2}(x)=1$ iff every line of $x$ has at least one $1$,}
\]
where \emph{lines} are rows and columns; hence, there are $2m$
lines. Examples of lowest ones $a\in \Low{f}$ of $f=\lines{n,2}$ in the
case $n=9$ is a permutation matrix
$a=\left(\begin{smallmatrix}0&1&0\\1&0&0\\0&0&1\end{smallmatrix}\right)$
but also matrices like
$a=\left(\begin{smallmatrix}0&1&1\\1&0&0\\1&0&0\end{smallmatrix}\right)$.
The function $\lines{n,2}$ is a special ($2$-dimensional) version of so-called ``blocking lines'' functions described in \cref{app:blocking-lines}.

The \emph{dual} of a Boolean function $f(x_1,\ldots,x_n)$
is\footnote{As before, for variables $x_i$, we write $\bar{x}_i$ instead of $\neg x_i$. } $\dual{f}(x_1,\ldots,x_n) :=\neg
f(\bar{x}_1,\ldots,\bar{x}_n)$. That is, we negate the input bits as
well as the obtained value. For example, by using DeMorgan rules
$\neg(x\lor y)=\neg x\land\neg y$ and $\neg(x\land y)=\neg x\lor\neg y$, we obtain that
the dual of
$f(x)=\bigwedge_{S\in\f}\big(\bigvee_{i\in S} x_i\big)$ is
$\dual{f}(x)=\bigvee_{S\in\f}\big(\bigwedge_{i\in S} x_i\big)$. That is,
$f(x)=1$ iff every set $S\in \f$ contains an $i\in S$ with $x_i=1$,
while $\dual{f}(x)=1$ iff there is a set $S\in \f$ with $x_i=1$ for
all $i\in S$.

Recall that $\sBool{f}{k}$ denotes the minimum size of a monotone \emph{syntactically} read-$k$  $(\lor,\land)$ circuit computing
a monotone Boolean function $f$.
In particular, $\Bool{f}{k}\leq \sBool{f}{k}$ always holds.

\begin{lem}\label{lem:gap}
  For the function $f=\lines{n,2}$, we have $\linBool{f}\geq
  \kMin{\Low{f}}{1}=\Bool{f}{1}=2^{\Omega(\sqrt{n})}$ but $\kMin{\Low{f}}{2}\leq
  \Bool{f}{2}\leq \sBool{f}{2}\leq 2n$ and $\Bool{\dual{f}}{1}\leq 2n$.
\end{lem}

\begin{proof}
  Let $A:=\Low{f}$ be the set of the lowest ones of the function
  $f=\lines{n,2}$.  The equality $\kMin{A}{1}=\Bool{f}{1}$ and the
  inequality $\kMin{A}{2}\leq \Bool{f}{2}$ are given by
  \cref{thm:trop-read-once}, while the inequality $\linBool{f}\geq
  \Bool{f}{1}$ is given by \cref{thm:multilinear}.

  To show the upper bound $\Bool{\dual{f}}{1}\leq 2m^2=2n$ for the dual $\dual{f}$ of the function $f=\lines{n,2}$, note that, for
  every input matrix $x$, we have
  $\dual{f}(x)=1$ iff $x$ has only $1$s on at least one line of $x$.
 In particular, prime implicants of $\dual{f}$ are the ANDs of variables
  corresponding to the $2m$ lines in the matrix~$x$. Hence, to obtain
  a read-$1$ circuit of size $\leq 2m^2=2n$ for $\dual{f}$ it is enough
  to take the OR of these ANDs.

  To show the lower bound $\Bool{f}{1}=2^{\Omega(\sqrt{n})}$, recall that
  $m\times m$ matrices $a\in f^{-1}(1)$ accepted by $f$ must have
  at least one $1$
  in each line (row or column).  None of such matrices can have fewer
  than $m$ $1$s, because otherwise, it would have an all-$0$ row or an
  all-$0$ column.  So, the smallest number of $1$s in a matrix $a\in
  A$ is $m$, and the matrices in $A$ with this number of $1$s are
  permutation matrices (with exactly one $1$ in each row and in each
  column). This means that the lower envelope $\lenv{A}$ of $A$ is the
  set $\Low{g}$ of the lowest ones of the perfect matching function
  $g=\match{m}$, and we already know that
  $\arithm{\Low{g}}=2^{\Omega(m)}$ holds
  (\cref{ex:matchings}). Together with \cref{thm:envel}, this yields
  $\Bool{f}{1}\geq \arithm{\lenv{A}}=\arithm{\Low{g}}=2^{\Omega(m)}$.

  To show the upper bound $\sBool{f}{2}\leq 2m^2=2n$, observe that $f$ can
  be computed by a trivial $(\lor,\land)$ circuit 
  \[
    F(x)=\bigwedge_{i=1}^m\big(\bigvee_{j=1}^n m_{i,j}\big)\land
    \bigwedge_{j=1}^m\big(\bigvee_{i=1}^m x_{i,j}\big)
  \]
  of size at most $2m^2$.  That is, we first compute the $2m$ ORs of
  variables along each line, and take the AND of these values. The
  arithmetic $(+,\times)$ version of $F$ produces the polynomial
  \[
  P(x)=\prod_{i=1}^m\big(\sum_{j=1}^m x_{i,j}\big)\cdot 
  \prod_{j=1}^m\big(\sum_{i=1}^m x_{i,j}\big)\,.
  \]
  Since no variable $x_{i,j}$ appears in this polynomial with a degree
  larger than $2$, the circuit $F$ is a syntactically
  read-$2$ circuit, as desired.
\end{proof}

\section{Concluding Remarks and Open Problems}
\label{sec:conclusions}

We have shown that already very restricted monotone Boolean
$(\lor,\land)$ circuits (read-$1$ circuits) capture the power of three
different types of circuits: monotone \emph{arithmetic} $(+,\times)$
circuits, \emph{tropical} $(\min,+)$ circuits, and non-monotone
Boolean \emph{multilinear} $(\lor,\land,\neg)$ circuits. The next
natural problem is to understand the power of read-$k$ $(\lor,\land)$
circuits for $k\geq 2$, with $k=2$ being the first nontrivial case.

It is clear that $\Bool{f}{1}\geq \Bool{f}{2}\geq \ldots\geq
\Bool{f}{k}\geq \ldots,\geq \BBool{f}$ holds for any monotone Boolean
function $f$, where $\BBool{f}$ is the minimum size of a monotone
$(\lor,\land)$ circuit computing~$f$. Super-polynomial lower bounds on
$\BBool{f}$ can be proved using the celebrated ``Method of
Approximations'' invented by Razborov~\cite{razb-clique,razb-perm,Razborov89}.
However, this method, as well as its later ``symmetric'' versions,
can be only applied to Boolean functions with a very special
combinatorial property: \emph{both} minterms and maxterms\footnote{A \emph{minterm} (resp., \emph{maxterm}) of
  a monotone Boolean function $f$ is a minimal under inclusion set of
  variables such that setting all these variables to $1$ (resp., to
  $0$) forces $f$ to output $1$ (resp., $0$) regardless of the values
  given to other variables. Since the function $f$ is monotone,
  every minterms intersects every maxterm.
  Note that prime implicants of $f$ are ANDs of all variables in minterms; ORs of all variables in maxterms are known as prime \emph{implicates} of~$f$; see, for example,~\cite{hammer}.}
  must be highly ``dispersed'' (not too
many of them can share a given number of variables in common).  For
example, already the application in~\cite{razb-perm} of the Method of
Approximations to prove the lower bound $\BBool{f}=n^{\Omega(\log n)}$
for the perfect matching function $f=\match{n}$ (which we considered
in \cref{sec:explicit}) is rather nontrivial, going deeply into the
structure of maxterms of this particular function: unlike the
minterms, the maxterms of $\match{n}$ are dispersed not highly enough
(see, e.g.,~\cite[Chapter~9]{myBFC-book} for more information).

In a sharp contrast, lower bounds
on the size of monotone read-$1$ circuits can be obtained \emph{without}
using the Method of Approximations: as demonstrated in \cref{sec:explicit},
an exponential lower bound $\Bool{f}{1}=2^{\Omega(n)}$ for  $f=\match{n}$
can be proved using a relatively simple argument: we had only to consider
the minterms of $f$.
 But what about read-$k$ circuits
for larger values of~$k$? In particular, what about read-$2$ circuits?

\begin{probl}\label{probl:1}
  Can super-polynomial lower bounds on the size of read-$2$ or at
  least of syntactically read-$2$ circuits be proved without using the
  Method of Approximations?
\end{probl}

\begin{table}[t]
  \begin{center}
    \begin{tabular}[t]{l@{\hskip 0.5cm}c@{\hskip 0.5cm}c@{\hskip
          0.5cm}c}
      \toprule
      Underlying semiring &  $\suma$-idempotence
      & $\daug$-idempotence & absorption\\
      $(R,\suma,\daug)$ & $x\suma x=x$
      & $x\daug x=x$ & $x\suma(x\daug y)=x$\\
      \midrule
      Arithmetic $(+,\times)$  &  {\large $-$}  & {\large $-$} & {\large $-$}\\[0.2em]
      Tropical $(\min,+)$  &  {\large $+$}  & {\large $-/+$} & {\large $+$}\\[0.2em]
      Read-$1$ $(\lor,\land)$ & {\large $+$} & {\large $-/+$} & {\large $+$} \\[0.2em]
      Multilinear $(\lor,\land)$ & {\large $+$} & {\large $-/+$} & {\large $+$} \\[0.2em]
      Tight $(\lor,\land)$ & {\large $+$} & {\large $+$} & {\large $-$} \\[0.2em]
      Unrestricted $(\lor,\land)$ & {\large $+$} & {\large $+$} & {\large $+$} \\
      \bottomrule
    \end{tabular}
  \end{center}
  \caption[]{\footnotesize Laws allowed $(+)$ or forbidden $(-)$ in various models of circuits. In read-$1$ and multilinear $(\lor,\land)$ circuits, as well as in tropical $(\min,+)$ circuits the usage of ``multiplicative'' idempotence  is only \emph{partially} forbidden $(-/+)$, because the usage of the absorption law $x\lor xy=x$ or, respectively, $\min\{x,x+y\}=x$ in these circuits is unrestricted; hence, the produced ``redundant'' terms can be eliminated using these laws.}
  \label{tab:laws}
\end{table}

Monotone Boolean read-$2$ circuits constitute the \emph{first} model
of computation---after tropical and monotone arithmetic
circuits---which can use \emph{both} the idempotence $x\land x=x$
\emph{and} the absorption $x\lor xy=x$ laws (albeit the usage of
idempotence is restricted). Let us stress that only
\emph{together} these two laws can unfold their full power.

Namely, the model where absorption  $x\lor xy=x$  is allowed (without any restriction), but (when producing prime implicants) idempotence
$x\land x=x$  is not allowed, is that
of read-$1$ circuits considered in this paper. We have seen
that exponential lower bounds for such circuits
\emph{can} be relatively easily proved
without using the Method of Approximations (\cref{ex:matchings}).

On the other hand, the model where idempotence $x\land x=x$ is allowed
(without any restriction), but absorption $x\lor xy=x$ is not allowed,
is that of so-called ``tight'' $(\lor,\land)$ circuits.  A monotone $(\lor,\land)$
circuit $F$ computing a Boolean function $f$ is \emph{tight} if the set
$B_F\subseteq\NN^n$ of exponent vectors of the formal $(+,\times)$ polynomial of $F$ satisfies the equality $\Supp{B_F}=\Supp{\Low{f}}$,
 not only the inclusions $\Supp{\Low{f}}\subseteq\Supp{B_F}$ and
 $B_F\subseteq\Up{(\Low{f})}$ as given by \cref{lem:struct-bool}.
That is, the circuit $F$ is tight if \emph{every} monomial of the
formal $(+,\times)$ polynomial of $F$ is a shadow of some prime
implicant of~$f$ (see \cref{tab:laws} for a schematic
comparison of various types of circuits).  Thus, tight $(\lor,\land)$ circuits cannot use the
absorption law $x\lor xy=x$, but (unlike in read-$k$ circuits) the
usage of idempotence law $x\land x=x$ is not restricted (degrees of
variables in the formal polynomial can be arbitrarily large). Note that
the read-$2$ $(\lor,\land)$ circuit used in the proof of
\cref{lem:gap} to compute the function $\lines{n,2}$ is tight. This shows
that tight circuits of degree already $2$ can be exponentially smaller
than (not necessarily tight) read-$1$ circuits.

Still, despite their alleged power, lower bounds for tight
$(\lor,\land)$ circuits (of arbitrary high degree) \emph{can} be
proved without using the Method of Approximations. This was demonstrated
in~\cite[Theorem~2]{juk-count}, where
a lower bound
$2^{\Omega(n)}$ on the size of tight $(\lor,\land)$ circuits computing
the perfect matching function $\match{n}$ is shown
using a fairly simple
argument similar to
that we used in \cref{sec:explicit} for read-$1$ circuits.
The point is that, because of the absence of the absorption $x\lor xy=x$, the
complexity of the function $\match{n}$ is also predetermined by the minterms of that function \emph{alone}. However, this argument fails if  the
absorption law is allowed.

In the case of read-$1$ circuits, we were able (in \cref{thm:envel}) to eliminate the
influence of absorption $x\lor xy=x$ by considering lower envelopes. But already in read-$2$ circuits, absorption can
(at least potentially) show its power.  So, a solution of
\cref{probl:1} could probably shed some light on where the power of
multiplicative idempotence $x\land x=x$ \emph{in combination with}
absorption $x\lor xy=x$ comes from.

The next natural question is: can larger allowed ``degree $k$ of
idempotence'' always substantially decrease the size of read-$k$
circuits? \Cref{lem:gap} shows that, for $k=1$, the gap
$\Bool{f}{k}/\Bool{f}{k+1}$ can be exponential.  But what about larger
values of~$k$?

\begin{probl}[Degree hierarchy]\label{probl:hierarchy}
  Can the gap $\Bool{f}{k}/\Bool{f}{k+1}$ or at least the gap
  $\Bool{f}{r}/\Bool{f}{k}$ be super-polynomial for $k\geq 2$ and
  $r$ not ``much'' smaller than $k$?
\end{probl}

To show such a gap, we need a function $f$ for which $\Bool{f}{r}$ is
``large'' but $\Bool{f}{k}$ is ``small.'' Hence,
\cref{probl:hierarchy} \emph{cannot} be solved using the Method of
Approximations because any lower bound on $\Bool{f}{r}$ obtained using this
method holds for \emph{every} $r$.

Yet another natural question is whether the gaps between the read-$k$ $(\lor,\land)$ circuit complexities of Boolean functions $f$ and their \emph{duals} $\dual{f}$ can be large.
\Cref{lem:gap} shows that, at least for $k=1$, the gap
$\Bool{f}{k}/\Bool{\dual{f}}{k}$
\emph{can} be large.

\begin{probl}[Duals]\label{probl:dual}
  Can the gap $\Bool{f}{k}/\Bool{\dual{f}}{k}$ be super-polynomial for
  all $k\geq 2$? In particular, can it be such for $k=2$?
\end{probl}
Note that also this question \emph{cannot} be answered using the
Method of Approximations because $\BBool{\dual{f}}=\BBool{f}$ always
holds: given a $(\lor,\land)$ circuit for $f$, we can obtain a
$(\lor,\land)$ circuit of the same size for the dual function
$\dual{f}$ by just interchanging AND and OR gates.

As mentioned in \cref{sec:readk}, the model of ``read-$k$ circuits'' is by analogy with the well-known computation model of ``read-$k$ times branching programs.''
So, let us briefly recall this latter model.
A (nondeterministic) \emph{branching program} (BP), also known as a \emph{switching-and-rectifier
network}, is a directed
acyclic graph, each edge of which is either a \emph{switch} (is
labeled by either a variable $x_i$ or by a negated variable $\bar{x}_i$)
or is a \emph{rectifier} (is labeled by constant $1$).
There is one node $s$ of zero indegree and one
node $t$ of zero outdegree. The term \emph{defined} by an $s$-$t$ path
is the AND of labels of its edges.  The Boolean function \emph{computed}
by a branching program is the OR of terms defined by all $s$-$t$ paths. The \emph{size} of
such a program is the total number of switches.  A variable $x_i$ is
\emph{read} along a path if $x_i$ or $\bar{x}_i$ appears as a label of
some edge along that path.
A branching program is a \emph{syntactically read-$k$} program
if no variable is read more than $k$ times along \emph{any} $s$-$t$
path (``syntactically'' because the restriction is on \emph{all} $s$-$t$ paths).

\markov

In \emph{monotone} branching programs, none of the edges is labeled by a negated variable $\bar{x}_i$.
By \cref{lem:struct-bool}, a monotone branching program $F$ computes a
monotone Boolean function $f$ iff the term defined by any $s$-$t$ path
is an implicant of $f$, and for every prime implicant $p$ of $f$ there
is an $s$-$t$ path in $F$ (a shadow path of $p$) along which only the
variables of $p$ are read. The program $F$
is a \emph{semantically read-$k$} BP if every
prime implicant of $f$ has \emph{at least one} shadow path along which no
variable is read more than $k$ times (see \cref{fig:markov} for an example). Thus, monotone semantically read-$k$ branching programs correspond to
read-$k$ $(\lor,\land)$ circuits considered in this paper: the restriction is only on shadow $s$-$t$ paths: there are no restrictions on the remaining $s$-$t$ paths.

For a monotone Boolean function $f$, let $\BP{f}{k}$ denote the minimum
number of switches in a monotone semantically read-$k$ branching
program computing~$f$. Due to the sequential nature of computation
in branching programs (rather than parallel
nature, as in the case of circuits), their structure
could be easier to analyze.

\begin{probl}[Read-$k$ branching programs]\label{probl:networks}
\Cref{probl:1,probl:hierarchy,probl:dual} but for monotone semantically read-$k$ branching programs instead of circuits, that is, for the measure $\BP{f}{k}$ instead of $\Bool{f}{k}$.
\end{probl}

As possible candidates for separating functions $f$ in
\cref{probl:hierarchy,probl:dual,probl:networks}, one could try so-called ``blocking
lines'' functions, including the functions  $\lines{n,k}$ and $\cov{n,k}$ described in \cref{ex:long-lines,ex:short-lines} of
\cref{app:blocking-lines} ($n$ stands for the number of variables of these functions). Each of these functions can be
computed by a monotone syntactically read-$k$ $(\lor,\land)$ circuit with $\leq kn$ gates, as well as by a monotone syntactically read-$k$ branching program with $\leq kn$ switches.

That the blocking lines function $\lines{n,2}$ exhibits a large read-$1$/read-$2$ gap is shown by \cref{lem:gap}.
That  blocking lines functions $\cov{n,k}$ \emph{can} exhibit large gaps
for read-$k$ branching programs even for larger parameters $k$
was shown by
Okolnishnikova~\cite{Okolnishnikova04}.
Namely, she has proved that if $k\geq 4$ is a constant and $1\leq r\leq \sqrt{k}$, then every (even not monotone but) \emph{syntactically}
read-$r$ branching program computing $\cov{n,k}$ must
have an exponential (in $n$) number of switches.
Using different (\emph{non}-monotone) functions, Thathachar~\cite{Thathachar98} has proved such a gap even for $r=k-1$.
But, to my best knowledge,
no similar gaps are known for monotone but \emph{semantically} read-$k$
branching programs (where the read-$k$ restriction is only on shadow $s$-$t$ paths).

\appendix

\section{Blocking lines functions}
\label{app:blocking-lines}

The Boolean function $\lines{n,2}$ we used in the
proof of \cref{lem:gap} is just a very special case of the following
more general construction of Boolean functions that have small read-$k$
circuits but could (apparently) require large read-$r$ circuits for
$r<k$. That these functions sometimes indeed \emph{can} exhibit such gaps is shown by \cref{lem:gap}, as well as by the aforementioned result of Okolnishnikova~\cite{Okolnishnikova04}.

Let $\Ln\subseteq 2^P$ be a family of subsets of a finite set $P$; let
us call elements $p\in P$ \emph{points}, and sets $L\in\Ln$
\emph{lines}.  Suppose that the family $\Ln$ is  $m$-\emph{uniform} (each
line has exactly $m$ points), and is $k$-\emph{regular} (each
point belongs to exactly $k$ lines). By double-counting, we have $m|\Ln|=k|P|$.  The \emph{blocking lines
  function} $f_{\Ln}$ has $n=|P|$ variables $x_p$, one for each point
$p\in P$ and, for every input $x\in\{0,1\}^P$, $f_{\Ln}(x)=1$ iff the set of points $S_x=\{p\in P\colon
x_p=1\}$ blocks (intersects) every line $L\in\Ln$. The monotone  circuit
$F_{\Ln}(x)=\bigwedge_{L\in\Ln}\big(\bigvee_{p\in L}x_p\big)$
 computes $f_{\Ln}$ and has $(m-1)|\Ln|+(|\Ln|-1)=m|\Ln|-1$ fanin-$2$
gates. Moreover, since no point belongs to more than $k$ lines,
$F_{\Ln}$ is a read-$k$ circuit. In particular, we have an upper bound
$\Bool{f_{\Ln}}{k}\leq m|\Ln|=k|P|=kn$. Note that for the dual
$\dual{f}_{\Ln}(x)=\bigvee_{L\in\Ln}\bigwedge_{p\in L}x_p$
of $f_{\Ln}$, we even have $\Bool{\dual{f}_{\Ln}}{1}\leq kn$, that is, the duals of blocking lines functions can be computed by small \emph{read-$1$} circuits. Since none of the points belongs to more
than $k$ lines, every implicant $t_S=\bigwedge_{p\in S} x_p$ of
$f_{\Ln}$ corresponds to a blocking set $S\subseteq P$ consisting of
$|S|\geq |\Ln|/k$ points, while shortest (prime) implicants $t_S$ of $f_{\Ln}$ (those with
the smallest number of variables) correspond to \emph{smallest}
blocking sets, that is, blocking sets $S\subseteq P$ consisting of
$|S|=|\Ln|/k$ pairwise noncollinear points; two points are
\emph{collinear} if they both belong to some line. Note that maxterms of $f_{\Ln}$ are sets $\{x_p\colon p\in L\}$ of variables corresponding to 
lines $L\in\Ln$, while minterms of $f_{\Ln}$ are sets $\{x_p\colon p\in S\}$ of variables corresponding to blocking sets $S\subseteq P$, none proper subset of which is a blocking set. 

Each point of a smallest blocking set $S\subseteq P$
blocks (intersects) its \emph{own} collection of $k$ lines: for every point $p\in S$ there is a collection $\Ln_p\subseteq\Ln$ of $|\Ln_p|=k$ lines such that $\Ln_p\cap\Ln_q=\emptyset$ holds for all points $p\neq q\in S$ (for, otherwise, points $p$ and $q$ would be collinear).
That is, each point $p\in S$ is the \emph{only} point of $S$ blocking the $k$ lines $\Ln_p$. In other words, each variable $x_p$ of a shortest prime implicant $t_S=\bigwedge_{p\in S} x_p$ of $f_{\Ln}$
blocks its \emph{own} collection of $k$ lines.
Intuitively, this means that the variable $x_p$ ``should'' be accessed by a circuit or a branching program at least $k$ times to produce this implicant.

\begin{ex}[Many short lines]\label{ex:short-lines}
  Instead of $2$-dimensional tensors (matrices), as in the case of the
  function $\lines{n,2}$ used in \cref{lem:gap}, one can consider $k$-dimensional tensors for
   $k\geq 3$. Let $n$ be of the form $n=m^k$.  As the
  underlying set $P$ of points, take the set $P=[m]^k$ of $k$-tuples
  $p=(p_1,\ldots,p_k)\in[m]^k$, and consider the family $\Ln\subseteq
  2^{P}$ of all (combinatorial) lines, where the \emph{line} in the
  $i$th direction through a point $p\in P$ is the set
  $L=\big\{(p_1,\ldots,p_{i-1},\ast,p_{i+1},\ldots,p_k)\colon
  \ast=1,\ldots,n\big\}$ of $|L|=m$ points.  Thus, we have
  $|\Ln|=km^{k-1}=(k/m)|P|$ distinct lines, no two sharing more than one point.
  Since each point belongs to exactly $k$ lines (there are $k$
  possible positions for $\ast$), the family $\Ln$ is
  $k$-regular. Hence, the corresponding (to this family $\Ln$)
  blocking lines function $\lines{n,k}(x):=f_{\Ln}(x)$ can be computed
  by a monotone read-$k$ $(\lor,\land)$ circuit of size $\leq m|\Ln|=k|P|=km^k=kn$; in particular,  $\Bool{\lines{n,k}}{k}\leq kn$ holds.
  Note that in this family $\Ln$ of lines, two points ($k$-tuples) $p\neq q\in [m]^k$ are collinear iff they differ in exactly one position. In particular, any set of \emph{pairwise} collinear points must entirely lie in one line.
  In the case
  $k=2$ (matrices, as in \cref{lem:gap}), smallest blocking sets consist of the entries of a permutation  matrix.
\qed\end{ex}

\begin{ex}[Few long lines]\label{ex:long-lines}
  Let $n$ be of the form $n=\binom{m}{k}$, were $m$ is divisible by
  $k$.  As the underlying set $P$ of points, take the collection
  $P=\binom{[m]}{k}$ of all $|P|=n$ $k$-element sets $p\subseteq [m]:=\{1,\ldots,m\}$.  For
  $i\in [m]$, let the line in the $i$th direction be the set
  $L_i=\{p\in P\colon i\in p\}$ of $|L_i|=\binom{m-1}{k-1}$ points
  containing~$i$. Let $\Ln=\{L_1, \ldots,L_m\}$ be the family of all $m$
  lines (in all $m$ directions). Since every point $p\in P$ consists
  of $k$ \emph{distinct} elements of $[m]$, each of them belongs to
  exactly $k$ lines; hence, the family $\Ln$ is $k$-regular For the
  corresponding (to this family $\Ln$) blocking lines function
  $\cov{n,k}(x):=f_{\Ln}(x)$, we have an upper bound
  $\Bool{\cov{n,k}}{k}\leq m|\Ln|=m
  \binom{m-1}{k-1}=k\binom{m}{k}=k|P|=kn$.
  In this family $\Ln$ of lines, two points $p\neq q\in \binom{[m]}{k}$ are collinear iff $p\cap q\neq \emptyset$, and a set $S\subseteq \binom{[m]}{k}$ of points ($k$-element subsets of $[m]$) is a blocking set iff the union of these subsets is the entire set $[m]$.
   Hence, smallest blocking sets $S\subseteq P$ consist of $|S|=m/k$ points forming a partition of
  $[m]$ into $m/k$ disjoint blocks of size $k$.
  In the case $k=2$, points $p\in P$ correspond to the edges $p=\{i,j\}$ of the complete graph
  $K_m$ on $[m]$, and the line $L_i$ in the $i$th direction is the set of
  all $m-1$ edges incident with vertex $i$. Then $\cov{n,2}(x)=1$ iff
  the subgraph $G_x$ of $K_m$ specified by $x$ has no isolated
  vertices. Smallest blocking sets in this case are  perfect matchings in~$K_m$.
\qed\end{ex}

\paragraph{Acknowledgments}
I am thankful to both referees for very useful comments and
suggestions.

\bibliographystyle{plain}

\footnotesize
\renewcommand{\baselinestretch}{0}

\end{document}